\documentclass{article}
\usepackage{amsmath, amsthm, amssymb,latexsym, gensymb}
\usepackage{cite}
\usepackage{geometry}
\usepackage{ifthen}
\usepackage{graphicx,psfrag, color, url}
\usepackage{bm} 
\usepackage{fancyhdr, hyperref}
 
\theoremstyle{plain}
\newtheorem{theorem}{Theorem}[section]
\newtheorem{lemma}[theorem]{Lemma}
\newtheorem{proposition}[theorem]{Proposition}

\newcommand{\hide}[1]{}

\begin{document}
\title{Crystallization for a Brenner-like potential}
\author{Brittan Farmer\\ {\small School of Mathematics, University of Minnesota}
\and Selim Esedo\={g}lu\\  {\small Department of Mathematics, University of Michigan}
\and Peter Smereka\thanks{Deceased}\\ {\small Department of Mathematics, University of Michigan}
}
\date{\today}

\pagestyle{plain}

\maketitle
\thispagestyle{fancy}
\begin{abstract}
	Graphene is a carbon molecule with the structure of a honeycomb lattice. We show how this structure can arise in two dimensions as the minimizer of an interaction energy with two-body and three-body terms. In the engineering literature, the Brenner potential is commonly used to describe the interactions between carbon atoms. We consider a potential of Stillinger-Weber type that incorporates certain characteristics of the Brenner potential: the preferred bond angles are 180 degrees and all interactions have finite range. We show that the thermodynamic limit of the ground state energy per particle is the same as that of a honeycomb lattice. We also prove that, subject to periodic boundary conditions, the minimizers are translated versions of the honeycomb lattice.
\end{abstract}

\section{Introduction}
Understanding why matter has a crystal structure at low temperature is a fundamental scientific question. At zero temperature, this problem can be treated mathematically by showing that the ground states of a given interaction potential have a periodic structure. This has been studied in one, two, and three dimensions for different choices of the potential. In this article, we consider a crystallization problem in two dimensions, building on earlier results by Heitmann and Radin \cite{heitmann80}, Radin \cite{radin81}, Theil \cite{theil06}, E and Li \cite{e09}, and Mainini and Stefanelli \cite{mainini14}. (Further references to results in one and three dimensions can be found in \cite{mainini14}.) In particular we wish to show that for an interaction potential consisting of a pair potential and a three-body term preferring 180\degree\ angles, a honeycomb lattice is the ground state. This problem is motivated by the use of the Brenner potential \cite{brenner90, brenner02} to describe carbon-carbon interactions. This potential was designed to model chemical bonding in small hydrocarbon molecules as well as graphite and diamond lattices \cite{brenner90}. The potential has been used to study the synthesis of carbon nanotubes \cite{maiti94, shibuta03, ding04, zhao05}, as well as the mechanics of carbon nanotubes \cite{yakobson96, zhang02}. Carbon in its graphene form has the structure of a honeycomb lattice, and the structure of carbon nanotubes is that of graphene wrapped around a cylinder. In this article, we connect the minimization of a Brenner-like potential in two dimensions to the formation of a honeycomb lattice structure.

The Brenner potential includes coupling between the dependence on the bond lengths and the bond angles. (See Appendix~\ref{sec:brenner} for more information about the Brenner potential.) However, we consider a potential of Stillinger-Weber type, which decouples the bond length and bond angle dependence into separate two-body and three-body terms. Nonetheless, our potential preserves certain features of the Brenner potential:
\begin{enumerate}
	\item the three-body contribution is minimized for bond angles of 180\degree, and
	\item the interactions have a cutoff distance beyond which they are zero.
\end{enumerate}

The formation of a honeycomb lattice by energy minimization in two dimensions has been studied previously by E and Li \cite{e09} and Mainini and Stefanelli \cite{mainini14}, under different assumptions on the potential. In both articles, the potential energy consists of a two-body term and a three-body term. In \cite{e09}, the two-body term includes long-distance interactions, whereas in \cite{mainini14}, it includes only first neighbor interactions. Long-distance interactions were included in \cite{e09} to allow potentials with Lennard-Jones-like decay. First neighbor interactions were used in \cite{mainini14} to represent the covalent bonds modeled by potentials like the Brenner potential. Our choice of first neighbor interactions in this article is motivated by the same reasons. As a result, we do not need to estimate any long-distance interactions, which was done in \cite{e09} and in earlier work by Theil \cite{theil06} using rigidity estimates, such as the one in \cite{friesecke02}.

In \cite{e09}, \cite{mainini14}, and the present work, the angular contributions to the three-body term are premultiplied by a cutoff function, so that contributions are only included if the central atom is less than a cutoff distance from the other two atoms. This cutoff radius is chosen to lie between the first nearest neighbors and second nearest neighbors in the honeycomb lattice. However, in both \cite{e09} and \cite{mainini14}, the three-body term prefers 120\degree\ bond angles, whereas our three-body term prefers 180\degree\ bond angles in order to match the angular contribution to the Brenner potential. This requires new arguments to explain the formation of a honeycomb lattice. In Lemma 2.6 of \cite{e09}, a uniform bond angle of 120\degree\ is enforced by taking the strength of the three-body interaction term high enough. However, if the three-body term is made sufficiently strong in our potential, then strings of atoms would have lower energy than the honeycomb lattice. The angular potential used by Mainini and Stefanelli \cite{mainini14} grows linearly from the minimum at 120\degree, unlike the quadratic Stillinger-Weber form used in \cite{e09}. In particular, this potential is non-differentiable at its minimum, which is a form of stickiness. In \cite{mainini14}, the authors also include conditions on the angular potential that ensure it is large for small angles, similar to our conditions \eqref{eq:V3_angle120} -- \eqref{eq:V3_angle72}. These assumptions suffice for them to prove that for a finite number of particles, the ground state of their potential is a subset of the honeycomb lattice. This is not true in general for our potential. For instance, in a system of three atoms, the configuration with a 180\degree\ bond angle would have a lower energy than the one with a 120\degree\ bond angle. The numerical studies of Kosimov et al. \cite{kosimov08, kosimov10} suggest that for the Brenner potential, this type of behavior is generic. For some cluster sizes, the ground state will be topologically equivalent to a honeycomb lattice, but for others, the ground states will contain a small number of defects. Even the topologically honeycomb structures are not true subsets of the honeycomb lattice, as the bond lengths and bond angles are not all equal, due to surface effects. Based on this work, it is reasonable to expect that for our Brenner-like potential, clusters of a variety of sizes will have ground states which are not subsets of the honeycomb lattice.

As a consequence, our results will be similar to those in Theil \cite{theil06} and E and Li \cite{e09}. In those works, long-range pair interactions made it possible that finite ground states could have a small number of defects, or relaxation of atomic positions at the boundary. Thus, they consider the minimal energy in an asymptotic limit, as well as the crystallization of infinite configurations subject to boundary conditions. We do the same in this article. However, the possibility of non-crystalline ground states is not due to long-distance interactions, as it was in \cite{theil06} and \cite{e09}, but instead due to the balance of short-range two-body and short-range three-body contributions. We obtain convergence of the per-particle ground-state energy to the same value as achieved by a honeycomb lattice, as the number of particles increases. If we enforce periodic boundary conditions, we obtain that ground states are honeycomb lattices. As in \cite{theil06} and \cite{e09}, our proofs involve obtaining lower and upper bounds on the ground state energy which, when normalized by $N$, converge to the same value as $N$ grows.

The structure of the article is as follows. In Section~\ref{sec:energy}, we define the potential function under consideration and give assumptions that will suffice to prove the formation of a honeycomb lattice. In Section~\ref{sec:estimates}, we prove an estimate on the per-particle energy in the thermodynamic limit and also give some results about the scaling of the excess surface energy. In Section~\ref{sec:honeycomb} we show that for periodic boundary conditions, the ground state of the energy is a honeycomb lattice.

\section{Energy of the system}
\label{sec:energy}
We consider a system of $N$ particles with index set $X_N$ and positions given by a map $y : X_N \rightarrow \mathbb{R}^2$. If $N$ is fixed, we will write $X = X_N$. We let $\theta_{x', x, x''}$ denote the angle from vector $y(x') - y(x)$ to vector $y(x'') - y(x)$, measured in the counterclockwise direction. See Figure~\ref{fig:angle}.

\begin{figure}[tb]
\centering
\includegraphics[width=.3\textwidth]{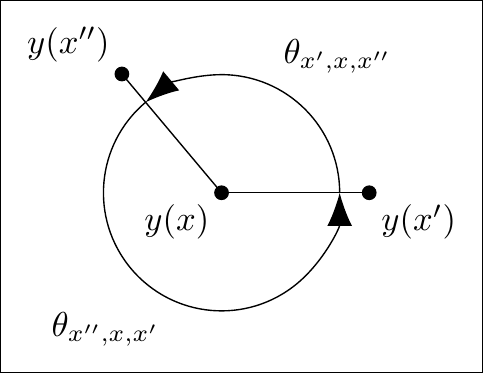}
\caption{Figure depicting our convention for $\theta_{x', x, x''}$.}
\label{fig:angle}
\end{figure}

The total energy of the system is:
\begin{align}
\label{eq:energy_full}
\begin{split}
	V( \{y\} ) = {} & \frac{1}{2} \sum_{x \in X} \sum_{\substack{x' \in X\\ x' \neq x}} V_2( |y(x') - y(x)| )\\
	& \quad + \frac{1}{2} \sum_{x \in X} \sum_{\substack{x', x'' \in X\\ x, x', x'' \text{ distinct}}} V_3 (\theta_{x', x, x''}) f( |y(x') - y(x)| ) f( |y(x'') - y(x)|).
\end{split}
\end{align}
The two-body term is the energy of the bond lengths between pairs of atoms. The three-body term is the energy of the bond angles formed by triples of atoms.

We will use the following short-hand for the energy of a pair or triple of atoms:
\begin{align*}
	e(\{ x, x' \}) &:= V_2( |y(x') - y(x)| ),\\
	a(x, x', x'') &:= V_3 (\theta_{x', x, x''}) f( |y(x') - y(x)| ) f( |y(x'') - y(x)|).
\end{align*}
Since the edge energy $e$ is independent of the ordering of the pair, we use the set notation for its argument. The factor of $\frac{1}{2}$ in the two-body term of \eqref{eq:energy_full} is present because each pair gets counted twice in the sum. For the angle energy $a$, the order of the arguments matters, as the first argument is the central atom in the bond angle and the second and third argument determine the direction in which we measure the angle. However, because our assumptions on $V_3(\theta)$ will imply that $V_3(\theta_{x'', x, x'}) = V_3(2\pi -  \theta_{x', x, x''}) = V_3(\theta_{x', x, x''})$, the angle energy $a$ is symmetric in the last two arguments, i.e. $a(x, x', x'') = a(x, x'', x')$, so each bond angle effectively gets counted twice in the three-body term in \eqref{eq:energy_full}. A factor of $\frac{1}{2}$ is used to accommodate for this.

We will make assumptions on our potentials which are compatible with the thermodynamic limit of the per-particle ground state energy being that of the honeycomb lattice. The (normalized) honeycomb lattice can be defined, as in \cite{e09}, as
\begin{equation*}
	H = \{ \xi = m a_1 + n a_2 + \ell b : m,n \in \mathbb{Z}, \ell = 0 \text{ or } 1 \},
\end{equation*}
where $a_1 = (\sqrt{3}, 0)$, $a_2 = \left( \frac{\sqrt{3}}{2}, \frac{3}{2} \right)$, and $b = (\sqrt{3}, 1)$. See Figure~\ref{fig:honeycomb}.
Note that in the honeycomb lattice, first-neighbors are distance 1 apart and second-neighbors are $\sqrt{3}$ apart.

\begin{figure}[tb]
\centering
\includegraphics[width=0.3\textwidth]{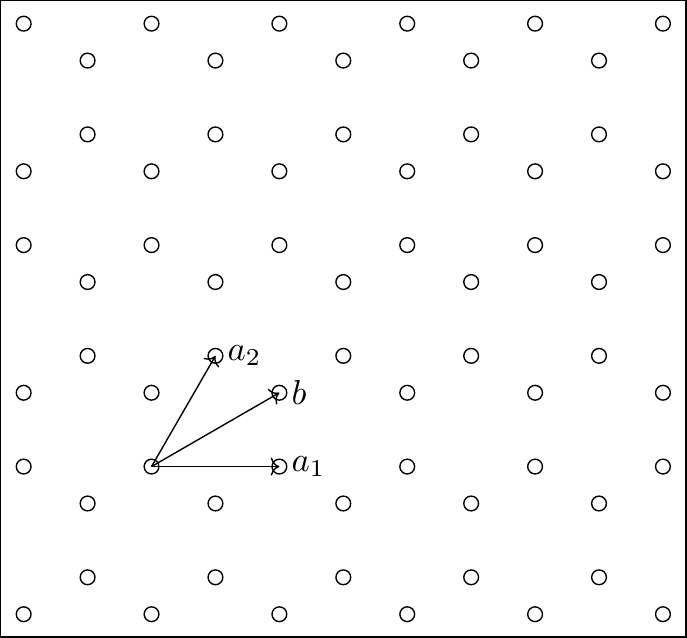}
\caption{A portion of the honeycomb lattice, with the basis vectors $a_1$ and $a_2$ and shift vector $b$.}
\label{fig:honeycomb}
\end{figure}

We make the following assumptions on the two-body potential. Let $0 < \alpha < \frac{1}{3}$ be a parameter. The two-body potential $V_2 = V_2(r; \alpha) : [0, \infty) \rightarrow (-\infty, \infty)$ satisfies the following assumptions:
\begin{align}
	V_2 &\in C^2(1 - \alpha, \infty), \label{eq:V2_smooth}\\
	V_2(r) &\geq \frac{1}{\alpha} \text{ for } r \leq 1-\alpha, \label{eq:V2_short_range}\\
	V_2''(r) &\geq 1\text{ for } r \in (1 - \alpha, 1 + \alpha), \label{eq:V2_convex}\\
	V_2(r) &\geq -\alpha \text{ for } r \in [1+\alpha, R), \label{eq:V2_mid_range}\\
	V_2(r) &= 0 \text{ for } r \geq R \text{, and} \label{eq:V2_long_range}\\
	\min_{r \geq 0} V_2(r) &= V_2(1) = -1 \label{eq:V2_normalized}.
\end{align}
Assumptions \eqref{eq:V2_smooth} -- \eqref{eq:V2_mid_range} were also used by Theil in \cite{theil06} (with $R = \frac{4}{3}$) and by E and Li \cite{e09} (with $R = \frac{3}{2}$). The pair potentials employed by Theil in \cite{theil06} and E and Li in \cite{e09} could be called Lennard-Jones-like, since their assumptions are compatible with the growth properties of the Lennard-Jones potential as $r \rightarrow 0$ and $r \rightarrow \infty$, but the well may be narrower than for the Lennard-Jones potential. Instead of the decay estimates on $V_2(r)$ used in \cite{theil06, e09}, we assume in \eqref{eq:V2_long_range} that there is zero interaction beyond a cutoff distance $R$. Thus, our pair potentials are similar to a truncated Lennard-Jones potential. As a result of the cutoff, we adopt the simpler normalization \eqref{eq:V2_normalized}.

We require that the cutoff $R$ satisfy $1 + \alpha < R < \sqrt{3}$ so that in a perfect honeycomb lattice, the only interactions are between first-neighbors. For the sake of definiteness, we choose $R = \frac{3}{2}$.\footnote{For the original Brenner potential \cite{brenner90}, the bond length for a honeycomb lattice having the lowest energy is either $1.42 \text{ \AA }$ or $1.45 \text{ \AA }$, depending on which of the two parameter sets is used. The Brenner potential is zero for interparticle distances greater than $2.0 \text{ \AA }$, which is about 1.4 times the optimal bond length for each of the two sets. This places the cutoff between the first and second neighbors in a perfect honeycomb lattice.} An example of a function $V_2$ satisfying assumptions \eqref{eq:V2_smooth} -- \eqref{eq:V2_normalized} is shown in Figure~\ref{fig:V2_example}. The parameter $\alpha$ plays several roles in determining the shape of the potential. As $\alpha$ decreases to zero, the close-range repulsion becomes stronger, the well becomes narrower, and the mid-range interaction becomes weaker.

\begin{figure}[tb]
\centering
\includegraphics[width=.4\textwidth]{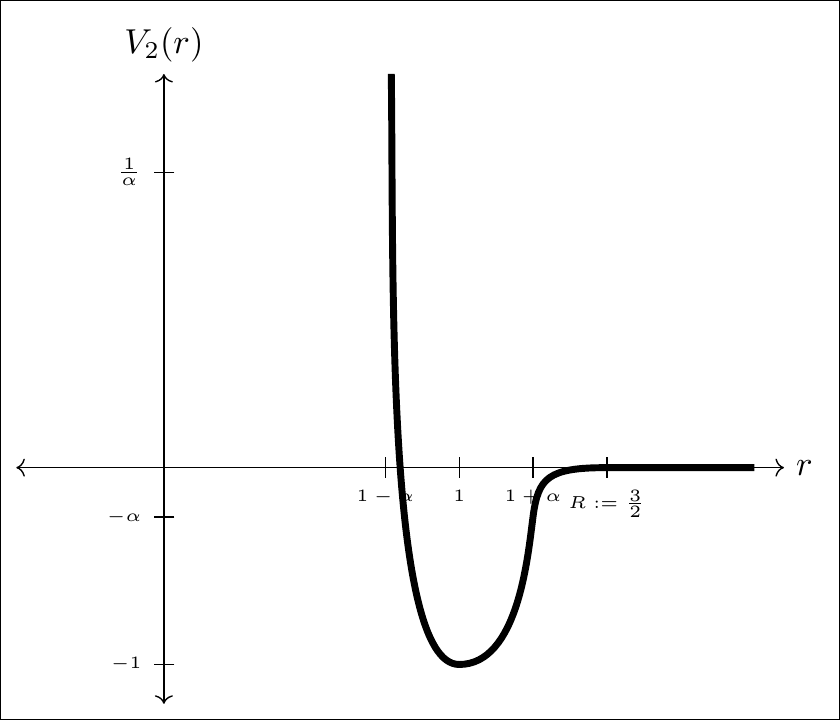}
\caption{Schematic of a function $V_2(r)$ satisfying the given assumptions.}
\label{fig:V2_example}
\end{figure}

The function $f$ in the three-body potential ensures that angular contributions are only included for triples where the last two particles are within a cutoff distance of the first. Various choices of this function are possible, but we use the cutoff function employed in the Brenner potential \cite{brenner90}. E and Li \cite{e09} and Mainini and Stefanelli \cite{mainini14} use cutoff functions that are different in detail, but share important qualitative features with ours (non-negativity, support on an interval bounded above by the second nearest neighbor distance, and strict positivity where the pair interaction is significant).

The cutoff function in the Brenner potential is parametrized by $R_1 < R_2$, which should both be chosen to be smaller than $\sqrt{3}$, so that in the honeycomb lattice, only bond angles between an atom and its first neighbors are included. We take $R_1 = \frac{3}{2}$ and $R_2 = \frac{5}{3}$. The cutoff function is:
\begin{equation*}
	f(r) =
	\begin{cases}
	1 & \text{if $r < R_1$}\\
	\frac{1}{2} \left( 1 + \cos( \pi (r - R_1)/(R_2 - R_1) ) \right) & \text{if $R_1 < r < R_2$}\\
	0 & \text{if $R_2 < r$}.
	\end{cases}
\end{equation*}
This $C^1(0,\infty)$ function has a smooth, monotonically decreasing transition from the value 1 to the value 0 on the interval $(R_1, R_2)$. Other forms of the cut-off function are possible for our results. Besides the generic properties mentioned in the previous paragraph, our proof of Lemma 3.12 requires that the cutoff be positive on the interval $[0, \frac{3}{2}]$, which is the entire interval were the pair interaction in non-zero. A different form of the cutoff function would require adjustments to some of the constants in our assumptions.

The angular potential $V_3 = V_3(\theta) : [0, 2\pi] \rightarrow [0, \infty)$ satisfies the following:
\begin{align}
	V_3 &\in C^2(0, 2\pi), \label{eq:V3_smooth}\\
	\text{there exists $m > 0$ such that } V_3''(\theta) &\geq m \text{ for all } \theta \in [0, 2\pi], \label{eq:V3_convex}\\
	V_3(\theta) &= V_3(2\pi - \theta) \text{ for all } \theta \in [0, 2\pi], \label{eq:V3_symmetric}\\
	\min_{\theta \in [0, 2\pi]} V_3(\theta) &= V_3(\pi) = 0. \label{eq:V3_normalized}
\end{align}
In our notation, the three-body interaction used by E and Li in \cite{e09} would correspond to $V_3(\theta) = \beta (\cos \theta + 1/2)^2$, with a different choice of cutoff function $f$. Assumptions~\eqref{eq:V3_smooth} and~\eqref{eq:V3_symmetric} would be satisfied by their potential as well. However, their angular function $V_3(\theta)$ attains a minimum of 0 at $\theta = \frac{2\pi}{3}$ instead of at $\theta = \pi$. The angular function employed by Mainini and Stefanelli in \cite{mainini14} is also minimized for angles of $\frac{2\pi}{3}$, but their function grows linearly out of the minimum.

Assumptions~\eqref{eq:V3_convex} and~\eqref{eq:V3_normalized} imply that the potential at different angles are ordered: $0 = V_3(\pi) < V_3(2\pi/3) < V_3(\pi/2) < V_3(2\pi/5) < V_3(\pi/3)$. Additionally, we suppose that
\begin{align}
	V_3(2\pi/3) &< \frac{1}{6}, \label{eq:V3_angle120}\\
	V_3(\pi/2) &> \frac{1}{8} + \frac{3}{4} V_3(2 \pi/3), \text{ and}\label{eq:V3_angle90}\\
	V_3(2\pi/5) &> 4. \label{eq:V3_angle72}
\end{align}
Note that assumptions~\eqref{eq:V3_angle120} and~\eqref{eq:V3_angle72} imply that
\begin{align}
	V_3(2\pi/5) >\frac{1}{5} + \frac{3}{5} V_3(2 \pi/3), \label{eq:V3_angle72b}
\end{align}
and assumption~\eqref{eq:V3_convex} implies that
\begin{align}
	V_3(\pi/3) > V_3(2\pi/5) > 4 > \frac{1}{4} + \frac{1}{2} V_3(2 \pi/3). \label{eq:V3_angle60}
\end{align}

Assumptions \eqref{eq:V3_angle120} -- \eqref{eq:V3_angle72} can be satisfied by taking, for example, $V_3(\frac{2\pi}{3}) = \frac{1}{8}$, $V_3(\frac{\pi}{2}) = \frac{1}{4}$, and $V_3(\frac{2\pi}{5}) = 5$. An example of a function $V_3$ satisfying these assumptions is shown in Figure~\ref{fig:V3_example}. Assumptions \eqref{eq:V3_angle120} -- \eqref{eq:V3_angle72}, particularly \eqref{eq:V3_angle72}, are similar to the assumption in \cite{mainini14} that $V_3 > 8$ on $(\theta_{\min} , \pi/2]$, where $\theta_{\min} := 2 \arcsin(1/(2\sqrt{2}))$. We need a greater number of assumptions since we make our estimates using different arguments.

\begin{figure}[tb]
\centering
\includegraphics[width=.5\textwidth]{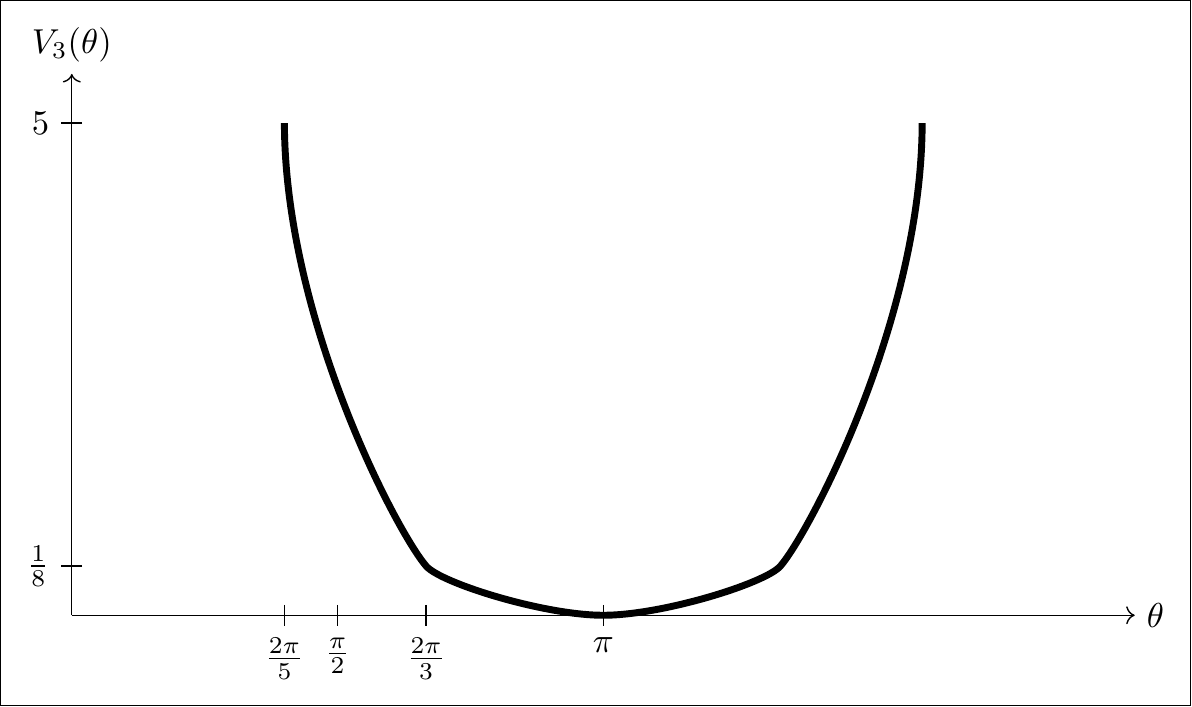}
\caption{Schematic of a function $V_3(\theta)$ satisfying the given assumptions.}
\label{fig:V3_example}
\end{figure}

\section{Estimates on the ground state energy}
Our main theorem states that the thermodynamic limit of the per-particle ground-state energy is the same as that of the honeycomb lattice $H$.
\label{sec:estimates}
\begin{theorem}
\label{thm:main}
There exists a constant $\alpha_0 \in (0, \frac{1}{3})$ such that for any $0 < \alpha < \alpha_0$, any potential of the form~\eqref{eq:energy_full}, with $V_2$ satisfying assumptions \eqref{eq:V2_smooth} -- \eqref{eq:V2_normalized} and $V_3$ satisfying assumptions \eqref{eq:V3_smooth} -- \eqref{eq:V3_angle72}, the following equation holds:
\begin{equation*}
	\lim_{N \rightarrow \infty} \min_{ y: X_N \rightarrow \mathbb{R}^2} \frac{1}{N} V(\{y\}) = -\frac{3}{2} + 3 V_3(2 \pi / 3).
\end{equation*}
\end{theorem}

One can prove that
\begin{equation}
\label{eq:energy_ub}
	\lim_{N \rightarrow \infty} \min_{ y: X_N \rightarrow \mathbb{R}^2} \frac{1}{N} V(\{y\}) \leq -\frac{3}{2} + 3 V_3(2 \pi / 3)
\end{equation}
by taking as a trial configuration a bijection from $X_N$ to the honeycomb lattice that does not create too much surface energy. See Proposition~\ref{prop:ub_Vmin} for a more quantitative upper bound on the ground state energy, based on the trial configurations constructed in \cite{mainini14}.

To prove that
\begin{equation*}
	\lim_{N \rightarrow \infty} \min_{ y: X_N \rightarrow \mathbb{R}^2} \frac{1}{N} V(\{y\}) \geq -\frac{3}{2} + 3 V_3(2 \pi / 3),
\end{equation*}
we will actually obtain a more complicated estimate, similar to those in \cite{theil06} and \cite{e09}, involving the minimum per-atom energy, the number of defects, and elastic corrections. We formulate this as a theorem.
\begin{theorem}
\label{thm:main_ineq}
There exists a constant $\alpha_0 \in (0, \frac{1}{3})$ and a constant $\Delta > 0$ such that for any ${0 < \alpha < \alpha_0}$ and any potential of the form~\eqref{eq:energy_full}, with $V_2$ satisfying assumptions \eqref{eq:V2_smooth} -- \eqref{eq:V2_normalized} and $V_3$ satisfying assumptions \eqref{eq:V3_smooth} -- \eqref{eq:V3_angle72}, and any ground state configuration $y: X_N \rightarrow \mathbb{R}^{\infty}$, there exists a subset $G \subset X_N$, containing only particles with three neighbors,\footnote{We shall later define the neighbors of particle $x \in X$ as those particles $x' \in X$ such that $|y(x') - y(x)| <  1+\alpha$.} such that the following inequality holds:
\begin{equation*}
\begin{split}
	V(\{y\}) &\geq
	\left( -\frac{3}{2} + 3 V_3(2\pi/3) \right) N + \frac{\Delta}{2} (N - \# G)\\
	& \qquad + \sum_{x \in G} \left( \sum_{j=1}^3 \frac{1}{2} \left| r_j(x) - 1 \right|^2 + \sum_{j=1}^3 \frac{m}{2} \left| \theta_j(x) - \frac{2\pi}{3} \right|^2 \right),
\end{split}
\end{equation*}
where $r_j(x)$ and $\theta_j(x)$ denote the bond lengths and bond angles between $x$ and its neighbors.
\end{theorem}

We will prove this theorem after obtaining intermediate results in the next several subsections. We first show that for ground states, the particles are well-separated. We will define two classes of atoms: regular and defected. Lower bounds on the neighborhood energy of arbitrary atoms and defected atoms are given. It turns out that regular atoms have only short-range interactions, i.e. they are not part of any pairs with mid-range distances between $1+\alpha$ and $\frac{3}{2}$. These results are combined to give the lower bound in Theorem~\ref{thm:main_ineq}. Together with the upper bound in \eqref{eq:energy_ub}, this proves Theorem~\ref{thm:main}.

\subsection{Estimates on the two- and three-body potentials and on the minimum inter-particle distance in the ground state}

First, for potentials $V_2$ and $V_3$ that satisfy our assumptions, we prove several estimates. The first is a quadratic lower bound on the pair potential in a neighborhood of its minimum.

\begin{lemma}[Estimate on potential $V_2$]
\label{lem:lb_V2}
	For a potential $V_2$ satisfying assumptions~\eqref{eq:V2_smooth},~\eqref{eq:V2_convex}, and~\eqref{eq:V2_normalized}, and $r \in (1 - \alpha, 1 + \alpha)$,
\begin{equation*}
	V_2(r) \geq -1 + \frac{1}{2} | r - 1 |^2.
\end{equation*}
\end{lemma}

\begin{proof}
	Since $V_2 \in C^2(1-\alpha, \infty)$ by assumption~\eqref{eq:V2_smooth} and it has a minimum of $-1$ at $r = 1$ by assumption \eqref{eq:V2_normalized}, we have $V_2'(1) = 0$. Since $V_2''(r) \geq 1$ on $(1 - \alpha, 1 + \alpha)$ by assumption \eqref{eq:V2_convex}, for $r \in (1 - \alpha, 1 + \alpha)$, we have
\begin{equation*}
	V_2(r) \geq V_2(1) + V_2'(1) (r-1) + \frac{1}{2} | r - 1 |^2 = -1 + 0 (r-1) + \frac{1}{2} | r - 1 |^2 = -1 + \frac{1}{2} | r - 1 |^2.
\end{equation*}
This proves the claim.
\end{proof}

Similarly, we also get a quadratic lower bound on the angular potential.

\begin{lemma}[Estimate on potential $V_3$]
\label{lem:lb_V3}
	For a potential $V_3$ satisfying assumptions~\eqref{eq:V3_smooth},~\eqref{eq:V3_convex}, and~\eqref{eq:V3_normalized}, and $\theta \in (0, 2\pi)$,
\begin{equation*}
	V_3(\theta) \geq \frac{m}{2} | \theta - \pi |^2.
\end{equation*}
\end{lemma}

\begin{proof}
	Since $V_3 \in C^2(0, 2\pi)$ by assumption~\eqref{eq:V3_smooth} and it has a minimum of $0$ at $\theta = \pi$ by assumption \eqref{eq:V3_normalized}, we have $V_3'(\pi) = 0$. Since $V_3''(\theta) \geq m$ on $(0, 2\pi)$ by assumption \eqref{eq:V3_convex}, for $\theta \in (0, 2\pi)$, we have
\begin{equation*}
	V_3(\theta) \geq V_3(\pi) + V_3'(\pi) (\theta-\pi) + \frac{m}{2} | \theta - \pi |^2 = \frac{m}{2} | \theta - \pi |^2.
\end{equation*}
This proves the claim.
\end{proof}

The next estimate gives a lower bound on the sum of the angular potential for angles that partition the circle. This sum is minimized when all the bond angles are equal.

\begin{lemma}[Estimate on sums of potential $V_3$]
\label{lem:lb_V3_sum}
	Let $M \in \mathbb{N}$ and suppose that $\bm{\theta} = (\theta_1, \theta_2, \ldots, \theta_M) \in [0, 2\pi]^M$ satisfies $\sum_{i=1}^M \theta_i = 2\pi$. Then, for a potential $V_3$ satisfying assumption~\eqref{eq:V3_convex},
\begin{equation*}
	\sum_{i=1}^M V_3(\theta_i) \geq M V_3 \left( \frac{2\pi}{M} \right) + \frac{m}{2} \sum_{i=1}^M \left| \theta_i - \frac{2\pi}{M} \right|^2.
\end{equation*}
\end{lemma}

\begin{proof}
	Define $g: [0, 2\pi]^M \rightarrow \mathbb{R}$ by $g(\bm{\theta}) = \sum_{i=1}^M V_3(\theta_i)$. We have $\nabla g(\bm{\theta}) = (V_3'(\theta_1), \ldots, V_3'(\theta_M))$ and $\nabla^2 g(\bm{\theta}) = \mathrm{diag} (V_3''(\theta_1), \ldots, V_3''(\theta_M))$. By assumption~\eqref{eq:V3_convex}, we have $V_3''(\theta) \geq m > 0$, implying that $\nabla^2 g(\bm{\theta}) \succeq mI$. Therefore, for any $\bm{\theta}_0 \in [0, 2\pi]$, we have
\begin{equation}
\label{eq:g_ineq}
	g(\bm{\theta}) \geq g(\bm{\theta}_0) + \nabla g(\bm{\theta}_0)^T (\bm{\theta} - \bm{\theta}_0) + \frac{m}{2} \| \theta - \theta_0 \|^2.
\end{equation}

Take the point $\bm{\theta}_0 = \frac{2\pi}{M} \left(1, \ldots, 1 \right)$. We have $g(\bm{\theta}_0) = M V_3(\frac{2\pi}{M})$ and $\nabla g(\bm{\theta}_0) = V_3'(\frac{2\pi}{M}) (1, \ldots, 1)$. Then,
\begin{align*}
	g(\bm{\theta}) &\geq g(\bm{\theta}_0) + \nabla g(\bm{\theta}_0)^T (\bm{\theta} - \bm{\theta}_0) + \frac{m}{2} \| \theta - \theta_0 \|^2\\
	&= M V_3 \left( \frac{2\pi}{M} \right) + V_3'\left( \frac{2\pi}{M} \right) (1, \ldots, 1)^T (\bm{\theta} - \bm{\theta}_0) + \frac{m}{2} \| \theta - \theta_0 \|^2\\
	&= M V_3 \left( \frac{2\pi}{M} \right) + V_3'\left( \frac{2\pi}{M} \right) \sum_{i=1}^M \left( \theta_i - \frac{2\pi}{M} \right) + \frac{m}{2} \sum_{i=1}^M \left| \theta_i - \frac{2\pi}{M} \right|^2\\
	&= M V_3 \left( \frac{2\pi}{M} \right) + V_3'\left( \frac{2\pi}{M} \right) \left( 2\pi - 2\pi \right) + \frac{m}{2} \sum_{i=1}^M \left| \theta_i - \frac{2\pi}{M} \right|^2\\
	&= M V_3 \left( \frac{2\pi}{M} \right) + \frac{m}{2} \sum_{i=1}^M \left| \theta_i - \frac{2\pi}{M} \right|^2.
\end{align*}

This proves the claim.
\end{proof}

Next, we prove an estimate on the minimum inter-particle distance in ground states, as was done in~\cite{theil06} and~\cite{e09}. We introduce the notation $B(z,r)$ for the ball of radius $r > 0$ centered at $z \in \mathbb{R}^2$.
\begin{lemma}
\label{lem:min_dist}
	There exists a constant $\alpha_0 \in \left( 0, \frac{1}{3} \right)$ such that for all $\alpha \in (0, \alpha_0)$ and all potentials $V$ of form~\eqref{eq:energy_full}, with $V_2$ satisfying assumptions \eqref{eq:V2_smooth} -- \eqref{eq:V2_normalized} and $V_3$ satisfying assumptions \eqref{eq:V3_smooth} -- \eqref{eq:V3_angle72}, all ground states $y: X \rightarrow \mathbb{R}^2$ of $V(\cdot)$ satisfy
\begin{equation}
\label{eq:min_dist}
	\min_{x \neq x'} |y(x') - y(x)| > 1 - \alpha.
\end{equation}
\end{lemma}

\begin{proof}
	This lemma actually follows from the more general Lemma 2.3 in \cite{e09}. The proof is slightly simpler for our finite-range potential, so we include it here.

Define $M$ to be the maximum number of particles contained in a ball of diameter $1 - \alpha$, i.e. $M := \max_{\eta \in \mathbb{R}^2} \# \{ x : y(x) \in B(\eta, \frac{1}{2} (1-\alpha)) \}$. We wish to show that $M=1$. We assume without loss of generality that $\eta = 0$. Define $B_M := B(0, \frac{1}{2} (1-\alpha))$ and $\mathcal{A} := y^{-1}(B_M)$.

We will consider the change in energy if the particles contained in $B_M$ are sent to infinity such that their mutual distances diverge. The change in energy, which must be non-negative since we are in a global minimizer, is
\begin{equation*}
	- \frac{1}{2} \sum_{\substack{x, x' \in \mathcal{A}\\ x \neq x'}} e(\{x, x'\}) - \sum_{x \in \mathcal{A}, x' \in X_N \setminus \mathcal{A}} e(\{x, x'\}) - \frac{1}{2} \sum_{\substack{x, x', x'' \in X_N,\\ \{x, x', x''\}  \cap \mathcal{A} \neq \emptyset\\ x, x', x'' \text{ distinct}}} a(x, x', x'') \geq 0.
\end{equation*}
So, we have
\begin{equation}
\label{eq:ineq1_min_dist}
	-\sum_{\substack{x, x' \in \mathcal{A}\\ x \neq x'}} e(\{x, x'\}) - \sum_{\substack{x, x', x'' \in X_N,\\ \{x, x', x''\}  \cap \mathcal{A} \neq \emptyset\\ x, x', x'' \text{ distinct}}} a(x, x', x'')
	\geq 2 \sum_{x \in \mathcal{A}, x' \in X_N \setminus \mathcal{A}} e(\{x, x'\}).
\end{equation}
For atoms closer together than $1- \alpha$, the energy is at least $\frac{1}{\alpha}$ by assumption \eqref{eq:V2_short_range}. Thus, we have
\begin{equation*}
\sum_{\substack{x, x' \in \mathcal{A}\\ x \neq x'}} e(\{x, x'\}) \geq \frac{1}{\alpha} M (M-1).
\end{equation*}
Also, since $f$ and $V_3$ are non-negative, we have
\begin{equation*}
\sum_{\substack{x, x', x'' \in X_N,\\ \{x, x', x''\}  \cap \mathcal{A} \neq \emptyset\\ x, x', x'' \text{ distinct}}} a(x, x', x'') \geq 0.
\end{equation*}
Thus, inequality \eqref{eq:ineq1_min_dist} becomes
\begin{equation}
\label{eq:ineq2_min_dist}
	-\frac{1}{\alpha} M (M-1) \geq 2 \sum_{x \in \mathcal{A}, x' \in X_N \setminus \mathcal{A}} e(\{x, x'\}).
\end{equation}
If $x \in \mathcal{A}$ and $x' \in X_N \setminus \mathcal{A}$, $e(\{x, x'\})$ will only be non-zero if $x'$ is in the annulus
\begin{equation*}
	B \left( 0, \frac{3}{2} + \frac{1}{2}(1- \alpha) \right) \setminus B \left( 0, \frac{1}{2}(1- \alpha) \right).
\end{equation*}
Since $0 < \alpha < \frac{1}{3}$, this region is contained within $B(0, 2) \setminus B \left( 0, \frac{1}{3} \right)$. There exists a constant $C > 0$, independent of $\alpha_0$, such that this region can be covered by $C$ balls of radius $\frac{1}{3}$, each of which can contain at most $M$ particles. The interaction energy $e$ is bounded below by $-1$, so we have
\begin{equation*}
	\sum_{\substack{x \in \mathcal{A}\\ x' \in X_N \setminus \mathcal{A}}} e(\{x, x'\}) \geq - \sum_{\substack{x \in \mathcal{A}\\ x' \in y^{-1}(B(0, 2) \setminus B(0, \frac{1}{3}))}} 1 \geq - M \cdot CM = -CM^2.
\end{equation*}
Therefore,
\begin{equation*}
	-\frac{1}{\alpha} M (M-1) \geq -2 C M^2,
\end{equation*}
which is equivalent to
\begin{equation*}
	(1 - 1/M) \leq 2 C \alpha.
\end{equation*}
Clearly, this inequality holds for any $\alpha > 0$ if $M = 1$. However, if $M \geq 2$, then
\begin{equation*}
	1/2 \leq (1 - 1/M) \leq 2 C \alpha.
\end{equation*}
If $\alpha < \frac{1}{4C}$, this inequality does not hold. Thus, taking $\alpha_0 = \frac{1}{4C}$, we must have $M=1$, which proves the result.
\end{proof}

\subsection{Definition of neighbors, regular atoms, and defected atoms}

For fixed $\alpha > 0$ and a configuration $y: X \rightarrow \mathbb{R}^2$, define the neighborhood of an atom as:
\begin{equation*}
	\mathcal{N}(x) := \{ x' \in X : x' \neq x \text{ and } |y(x') - y(x)| \leq 1 + \alpha \}.
\end{equation*}

Next, we show that for configurations satisfying an appropriate minimum distance property, all particles have at most six neighbors, as in \cite{theil06}.\footnote{Our definition differs from Theil's in that $x \notin \mathcal{N}(x)$.}

\begin{lemma}
\label{lem:n_nbrs}
	There exists a constant $\alpha_0 \in (0, \frac{1}{3})$ such that for all $\alpha$ with $0 < \alpha < \alpha_0$ and all configurations $y : X \rightarrow \mathbb{R}^2$ which satisfy \eqref{eq:min_dist}, we have
\begin{equation*}
	\# \mathcal{N}(x) \leq 6 \text{ for all } x \in X.
\end{equation*}
\end{lemma}
\begin{proof}
	As mentioned in \cite{theil06}, the proof follows from Lemma~\ref{lem:min_dist} by a geometric argument.
\end{proof}

We define the \emph{neighborhood energy} as
\begin{equation*}
	V_{\mathcal{N}} (x) := \frac{1}{2} \sum_{x' \in \mathcal{N}(x)} e(\{ x, x' \})
	+ \frac{1}{2} \sum_{\substack{x', x'' \in \mathcal{N}(x)\\ x' \neq x''}} a(x, x', x'').
\end{equation*}
Note that only half of the energy of a pairwise bond is associated with the atom $x$.

We define a subset of the particles which have three neighbors and angles close to $2\pi/3$. We will see later that these particles have a neighborhood energy very close to the optimal one. We call these \emph{regular atoms}:
\begin{equation*}
	G_{\epsilon} = \{ x \in X_N \ |\ \# \mathcal{N}(x) = 3 \text{ and } | \theta_{x', x, x''} - 2\pi/3 | < \epsilon \text{ for all } x', x'' \in \mathcal{N}(x) \}.
\end{equation*}

The choice of regular atoms is parametrized by $\epsilon > 0$, which we shall later take to be a fixed value. Atoms which are not regular will be called \emph{defected}.

\subsection{Decomposition of the energy}
\label{sec:energy_decomp}

With these definitions, we decompose and estimate the energy.
\begin{lemma}
\label{lem:energy_decomp}
For any configuration $y: X \rightarrow \mathbb{R}^2$ and any potential of the form~\eqref{eq:energy_full}, with $V_3$ satisfying assumption \eqref{eq:V3_normalized}, we have
\begin{align}
\label{eq:energy_decomp}
	V(\{ y \}) &\geq \sum_{x \in G_\epsilon} V_{\mathcal{N}} (x)
	+ \sum_{x \notin G_\epsilon} V_{\mathcal{N}} (x)
	+ \frac{1}{2} \sum_{x \in X}\ \sum_{\substack{x' \notin \mathcal{N}(x)\\ x' \neq x}} e(\{ x, x' \}).
\end{align}
\end{lemma}

\begin{proof}
The potential has the form
\begin{align}
\label{eq:energy}
	V &= \frac{1}{2} \sum_{x \in X}\ \sum_{\substack{x' \in X\\ x' \neq x}} e(\{ x, x' \})
	+ \frac{1}{2} \sum_{x \in X}\ \sum_{\substack{x', x'' \in X\\ x, x', x'' \text{ distinct}}} a(x, x', x'').
\end{align}

First, we split the sum over pairs into sums over neighboring pairs and non-neighboring pairs:
\begin{align}
\label{eq:pair_decomp}
	\sum_{x \in X}\ \sum_{\substack{x' \in X\\ x' \neq x}} e(\{ x, x' \}) &=
	\sum_{x \in X}\ \sum_{x' \in \mathcal{N}(x)} e(\{ x, x' \})
	+ \sum_{x \in X}\ \sum_{\substack{x' \notin \mathcal{N}(x)\\ x' \neq x}} e(\{ x, x' \}).
\end{align}

Now, we split up the sum over triples in a similar way and estimate:
\begin{align}
	\sum_{x \in X}\ \sum_{\substack{x', x'' \in X\\ x, x', x'' \text{ distinct}}} a(x, x', x'') &=
	\sum_{x \in X}\ \sum_{\substack{x', x'' \in \mathcal{N}(x)\\ x' \neq x''}} a(x, x', x'') 
	+ \sum_{x \in X}\ \sum_{\substack{ \{x', x''\} \cap \mathcal{N}^c(x) \neq \emptyset \\ x, x', x'' \text{ distinct}}} a(x, x', x'') \\
	\label{eq:triple_decomp}
	& \geq \sum_{x \in X}\ \sum_{\substack{x', x'' \in \mathcal{N}(x)\\ x' \neq x''}} a(x, x', x''),
\end{align}
where we have used assumption~\eqref{eq:V3_normalized} that the angle energy is non-negative.

Now, substituting \eqref{eq:pair_decomp} and \eqref{eq:triple_decomp} into \eqref{eq:energy}, we have
\begin{align*}
	V &\geq
	\frac{1}{2} \sum_{x \in X}\ \sum_{x' \in \mathcal{N}(x)} e(\{ x, x' \})
	+ \frac{1}{2} \sum_{x \in X}\ \sum_{\substack{x' \notin \mathcal{N}(x)\\ x' \neq x}} e(\{ x, x' \})
	+ \frac{1}{2} \sum_{x \in X}\ \sum_{\substack{x', x'' \in \mathcal{N}(x)\\ x' \neq x''}} a(x, x', x'')\\
	&= \sum_{x \in X} \left[ \frac{1}{2} \sum_{x' \in \mathcal{N}(x)} e(\{ x, x' \})
	+ \frac{1}{2} \sum_{\substack{x', x'' \in \mathcal{N}(x)\\ x' \neq x''}} a(x, x', x'') \right]
	+ \frac{1}{2} \sum_{x \in X}\ \sum_{\substack{x' \notin \mathcal{N}(x)\\ x' \neq x}} e(\{ x, x' \})\\
	&= \sum_{x \in X} V_{\mathcal{N}} (x)
	+ \frac{1}{2} \sum_{x \in X}\ \sum_{\substack{x' \notin \mathcal{N}(x)\\ x' \neq x}} e(\{ x, x' \}),
\end{align*}
where we have used the definition of the neighborhood energy. Now, splitting the sum of the neighborhood energies into sums over regular and defected atoms, we have
\begin{align*}
	V &\geq \sum_{x \in G_\epsilon} V_{\mathcal{N}} (x)
	+ \sum_{x \notin G_\epsilon} V_{\mathcal{N}} (x)
	+ \frac{1}{2} \sum_{x \in X}\ \sum_{\substack{x' \notin \mathcal{N}(x)\\ x' \neq x}} e(\{ x, x' \}).
\end{align*}
This proves the claim.
\end{proof}

\subsection{Estimates on the neighborhood energy}
\label{sec:nbhd_energy_est}

In this section, we consider a fixed $x \in X$ and prove several estimates on its neighborhood energy. First, we note that for $\alpha \in (0, \frac{1}{3})$, for any configuration $y: X \rightarrow \mathbb{R}^2$ satisfying the minimum distance inequality~\eqref{eq:min_dist}, we can enumerate $\mathcal{N}(x)$ in counterclockwise order around $x$, starting from the $e_1$-direction: $\mathcal{N}(x) = \{ x_i \}_{i=1}^{M}$, where $M = \# \mathcal{N}(x)$.\footnote{We can index the neighbors in this way because no two can be along the same ray. Assuming this were possible, the minimum distance property would imply that the further particle is at least distance $2 - 2 \alpha$ from $x$. For $\alpha < \frac{1}{3}$, we have $2 - 2 \alpha > 1 + \alpha$, implying that it is in fact not a neighbor, contradicting the original assumption.} We make the convention that $x_{M+1} = x_1$. We label the bond lengths $r_i := |y(x_i) - y(x)|$ and the consecutive angles $\theta_i := \theta_{x_i, x, x_{i+1}} = 2\pi - \theta_{x_{i+1}, x, x_{i}}$ for $i = 1, \ldots, M$. These angles satisfy the constraint $\sum_{i=1}^M \theta_i = 2\pi$. Recall that assumption~\eqref{eq:V3_symmetric} implies that $V_3(\theta_{x_i, x, x_{i+1}}) = V_3(\theta_{x_{i+1}, x, x_{i}})$.

We now prove a lower bound on the neighborhood energy of an arbitrary particle.

\begin{lemma}
\label{lem:lb_all}
	There exists $\alpha_0 > 0$ such that for all $\alpha \in (0, \alpha_0)$, for all potentials $V_2$ satisfying assumptions \eqref{eq:V2_smooth} -- \eqref{eq:V2_normalized} and potentials $V_3$ satisfying assumptions \eqref{eq:V3_convex} -- \eqref{eq:V3_angle72}, and all configurations $y: X \rightarrow \mathbb{R}^2$ satisfying \eqref{eq:min_dist}, the neighborhood energy $V_{\mathcal{N}}(x)$ is bounded below for all $x \in X$.
	
	Fix $x \in X$ and let $M := \# \mathcal{N}(x)$. By Lemma~\ref{lem:n_nbrs}, there exists $\alpha_0$ such that $M \leq 6$. For $0 \leq M \leq 1$, we have
\begin{equation*}
	V_{\mathcal{N}} (x) \geq -\frac{M}{2} + \frac{1}{4} \sum_{i = 1}^{M} | r_i - 1 |^2.
\end{equation*}

For $M = 2$, we have
\begin{equation*}
	V_{\mathcal{N}} (x) \geq -\frac{M}{2} + \frac{1}{4} \sum_{i = 1}^{M} | r_i - 1 |^2 + \frac{m}{2} \left| \theta_1 - \pi \right|^2.
\end{equation*}

For $3 \leq M \leq 6$, we have
\begin{equation*}
	V_{\mathcal{N}} (x) \geq -\frac{M}{2} + \frac{1}{4} \sum_{i = 1}^{M} | r_i - 1 |^2 + M V_3 \left( \frac{2\pi}{M} \right) + \frac{m}{2} \sum_{i=1}^M \left| \theta_i - \frac{2\pi}{M} \right|^2.
\end{equation*}
\end{lemma}

\begin{proof}
Let $x \in X$ and $M = \# \mathcal{N}(x)$. The neighborhood energy is
\begin{equation*}
	V_{\mathcal{N}} (x) = \frac{1}{2} \sum_{x' \in \mathcal{N}(x)} e(\{ x, x' \})
	+ \frac{1}{2} \sum_{\substack{x', x'' \in \mathcal{N}(x)\\ x' \neq x''}} a(x, x', x'').
\end{equation*}

First, we estimate the edge energy. For $x_i \in \mathcal{N}(x)$, by definition of $\mathcal{N}(x)$ and by equation~\eqref{eq:min_dist}, we have $|y(x_i) - y(x)| \in (1 - \alpha, 1 + \alpha)$. Therefore, using Lemma~\ref{lem:lb_V2}, we have
\begin{align}
	\nonumber
	e(\{x, x_i\}) = V_2(|y(x_i) - y(x)|) \geq -1 + \frac{1}{2} | r_i - 1 |^2, \text{ and}\\
	\label{eq:lb_edge}
	\frac{1}{2} \sum_{x' \in \mathcal{N}(x)} e(\{ x, x' \}) \geq -\frac{M}{2} + \frac{1}{4} \sum_{i = 1}^{M} | r_i - 1 |^2
\end{align}

Now, we treat the angle energy. Note that for any $x' \in \mathcal{N}(x)$, $|y(x') - y(x)| \leq 1+\alpha < \frac{3}{2}$, so by definition $f( |y(x') - y(x)| ) = 1$. Thus, for $x', x'' \in \mathcal{N}(x)$,
\begin{equation}
	a(x, x', x'') = V_3 (\theta_{x', x, x''}) f( |y(x') - y(x)| ) f( |y(x'') - y(x)|) = V_3 (\theta_{x', x, x''}).
\end{equation}

If $0 \leq M \leq 1$, there are no triples, and
\begin{equation}
	\label{eq:lb_angle_01}	
	\frac{1}{2} \sum_{\substack{x', x'' \in \mathcal{N}(x)\\ x' \neq x''}} a(x, x', x'') = 0.
\end{equation}

If $M = 2$, there are triples $(x, x_1, x_2)$ and $(x, x_2, x_1)$. Therefore,
\begin{equation}
	\label{eq:lb_angle_2}
	\frac{1}{2} \sum_{\substack{x', x'' \in \mathcal{N}(x)\\ x' \neq x''}} a(x, x', x'') = \frac{1}{2} ( V_3 (\theta_{x_1, x, x_2}) + V_3 (\theta_{x_2, x, x_1}) ) = \frac{1}{2} ( 2 V_3(\theta_1) ) = V_3(\theta_1) \geq \frac{m}{2} |\theta_1 - \pi|^2.
\end{equation}

If $3 \leq M \leq 6$, then the distinct triples $(x, x_{i}, x_{i+1})$ and $(x, x_{i+1}, x_i)$ are contained in the set of neighboring triples for $i=1, \ldots, M$. Since $V_3$ is nonnegative by assumption~\eqref{eq:V3_normalized}, we have
\begin{align*}
	\frac{1}{2} \sum_{\substack{x', x'' \in \mathcal{N}(x)\\ x' \neq x''}} a(x, x', x'')
	&\geq \sum_{i=1}^M \frac{1}{2} \left( a(x, x_{i}, x_{i+1}) + a(x, x_{i+1}, x_{i}) \right) \\
	&= \sum_{i=1}^M \frac{1}{2} \left( V_3 (\theta_{x_i, x, x_{i+1}}) + V_3 (\theta_{x_{i+1}, x, x_i}) \right)\\
	&= \sum_{i=1}^M V_3(\theta_i).
\end{align*}

Therefore, since $\sum_{i=1}^M \theta_i = 2\pi$, using Lemma~\ref{lem:lb_V3_sum},
\begin{align}
	\label{eq:lb_angle_36}
	\frac{1}{2} \sum_{\substack{x', x'' \in \mathcal{N}(x)\\ x' \neq x''}} a(x, x', x'')
	&\geq M V_3 \left( \frac{2\pi}{M} \right) + \frac{m}{2} \sum_{i=1}^M \left| \theta_i - \frac{2\pi}{M} \right|^2.
\end{align}

Finally, we combine the estimates on the edge energy and the angle energy to estimate the neighborhood energy. For $0 \leq M \leq 1$, by Equations~\eqref{eq:lb_edge} and~\eqref{eq:lb_angle_01}, we have
\begin{equation*}
	V_{\mathcal{N}} (x) \geq -\frac{M}{2} + \frac{1}{4} \sum_{i = 1}^{M} | r_i - 1 |^2.
\end{equation*}

For $M = 2$, by Equations~\eqref{eq:lb_edge} and~\eqref{eq:lb_angle_2}, we have
\begin{equation*}
	V_{\mathcal{N}} (x) \geq -\frac{M}{2} + \frac{1}{4} \sum_{i = 1}^{M} | r_i - 1 |^2 + \frac{m}{2} |\theta_1 - \pi|^2.
\end{equation*}

For $3 \leq M \leq 6$, by Equations~\eqref{eq:lb_edge} and~\eqref{eq:lb_angle_36} we have
\begin{equation*}
	V_{\mathcal{N}} (x) \geq -\frac{M}{2} + \frac{1}{4} \sum_{i = 1}^{M} | r_i - 1 |^2 + M V_3 \left( \frac{2\pi}{M} \right) + \frac{m}{2} \sum_{i=1}^M \left| \theta_i - \frac{2\pi}{M} \right|^2. \quad
\end{equation*}
\end{proof}

We recognize the quadratic terms in Lemma~\ref{lem:lb_all} as elastic terms. Therefore, we introduce the notation
\begin{align*}
	W_{e}(x) &= \frac{1}{4} \sum_{j=1}^{\# \mathcal{N}(x)} \left| r_j - 1 \right|^2 \text{ and}\\
	W_{a}(x) &= \frac{m}{2} \sum_{j=1}^{\# \mathcal{N}(x)} \left| \theta_j - \frac{2\pi}{\# \mathcal{N}(x)} \right|^2.
\end{align*}
Note that by their definitions, $W_{e}(x)$ and $W_{a}(x)$ are non-negative.

Rewriting the estimates in Lemma~\ref{lem:lb_all} using this notation, we have
\begin{align*}
	V_{\mathcal{N}} (x) & \geq -\frac{M}{2} + W_{e}(x) & & \text{for $0 \leq M \leq 1$,}\\
	V_{\mathcal{N}} (x) & \geq -\frac{M}{2} + W_{e}(x) + \frac{m}{2}|\theta_1 - \pi|^2 & & \text{for $M = 2$, and}\\
	V_{\mathcal{N}} (x) & \geq -\frac{M}{2} + M V_3 \left( \frac{2\pi}{M} \right) + W_{e}(x) + W_{a}(x) & & \text{for $3 \leq M \leq 6$.}
\end{align*}

We now obtain a lower bound on the neighborhood energy of a defected particle.

\begin{lemma}
\label{lem:lb_defected}
	There exists $\alpha_0 > 0$ such that for all $\alpha \in (0, \alpha_0)$, for all $\epsilon > 0$, for all potentials $V_2$ satisfying assumptions \eqref{eq:V2_smooth} -- \eqref{eq:V2_normalized} and potentials $V_3$ satisfying assumptions \eqref{eq:V3_smooth} -- \eqref{eq:V3_angle72}, and for all configurations $y: X \rightarrow \mathbb{R}^2$ satisfying~\eqref{eq:min_dist}, there exists $\Delta_\epsilon > 0$ such that
\begin{equation*}
	V_{\mathcal{N}}(x) \geq -\frac{3}{2} + 3 V_3(2 \pi/3) + \Delta_\epsilon \text{ for all } x \notin G_{\epsilon}.
\end{equation*}
\end{lemma}

\begin{proof}
If $x \notin G_\epsilon$, then either $M \neq 3$ or $M = 3$ and $|\theta_i - \frac{2\pi}{3}| \geq \epsilon$ for at least one $i$.

\textbf{Case 1.} Suppose $M \neq 3$. By Lemma~\ref{lem:lb_all}, for $0 \leq M \leq 2$, we have
\begin{equation*}
	V_{\mathcal{N}} (x) \geq -\frac{M}{2} + W_{e}(x) \geq -\frac{M}{2} \geq -1.
\end{equation*}
	Also by Lemma~\ref{lem:lb_all}, for $4 \leq M \leq 6$, we have
\begin{equation*}
	V_{\mathcal{N}} (x) \geq -\frac{M}{2} + M V_3 \left( \frac{2\pi}{M} \right) + W_{e}(x) + W_{a}(x)
	\geq -\frac{M}{2} + M V_3 \left( \frac{2\pi}{M} \right).
\end{equation*}

Therefore, for $M \neq 3$, we have
\begin{equation}
\label{eq:nbhd_neq3}
    	V_{\mathcal{N}}(x) \geq \min \left\{ -1, -2 + 4 V_3(\pi/2), -\frac{5}{2} + 5 V_3(2 \pi/5),  -3 + 6 V_3(\pi/3) \right\}.
\end{equation}
Define
\begin{equation*}
    	\Delta_b := \min \left\{ -1, -2 + 4 V_3(\pi/2), -\frac{5}{2} + 5 V_3(2 \pi/5),  -3 + 6 V_3(\pi/3) \right\} - \left( -\frac{3}{2} + 3 V_3(2 \pi/3) \right).
    \end{equation*}
    The assumptions~\eqref{eq:V3_angle120} -- \eqref{eq:V3_angle72} and equations~\eqref{eq:V3_angle72b} and~\eqref{eq:V3_angle60} imply that $\Delta_b > 0$. (Recall that~\eqref{eq:V3_angle72b} and~\eqref{eq:V3_angle60} followed from assumptions~\eqref{eq:V3_smooth} -- \eqref{eq:V3_angle72}.) By definition of $\Delta_b$ and Equation~\eqref{eq:nbhd_neq3}, we have
\begin{equation*}
    	V_{\mathcal{N}}(x) \geq -\frac{3}{2} + 3 V_3(2 \pi/3) + \Delta_b.
\end{equation*}
    
\textbf{Case 2.} Suppose $M = 3$ and $|\theta_i - \frac{2\pi}{3}| \geq \epsilon$ for at least one $i$. Then by Lemma~\ref{lem:lb_all},
    \begin{equation*}
	V_{\mathcal{N}} (x) \geq -\frac{3}{2} + 3 V_3 \left( \frac{2\pi}{3} \right) + W_{e}(x) + \frac{m}{2} \sum_{i=1}^M \left| \theta_i - \frac{2\pi}{M} \right|^2 \geq -\frac{3}{2} + 3 V_3 \left( \frac{2\pi}{3} \right) + \frac{m}{2} \epsilon^2.
    \end{equation*}
    
    Now, define $\Delta_\epsilon := \min \{ \Delta_b, \frac{m}{2} \epsilon^2 \}$. If $x \notin G_{\epsilon}$, then $V_{\mathcal{N}}(x) \geq -\frac{3}{2} + 3 V_3(2 \pi/3) + \Delta_\epsilon$.
\end{proof}

\subsection{Estimates on the mid-range interactions}
We now estimate the sum of edge energies over pairs which are not neighbors. First, denote the set of mid-range pairs as
\begin{equation*}
	\mathcal{M} := \left\{ \{x, x'\} \subset X : 1+\alpha < |y(x) - y(x')| < \frac{3}{2} \right\}.
\end{equation*}

\begin{lemma}
\label{lem:mid-range}
For any potential $V_2$ satisfying assumptions~\eqref{eq:V2_mid_range} and~\eqref{eq:V2_long_range},
\begin{align*}
	\frac{1}{2} \sum_{x \in X} \sum_{\substack{x' \notin \mathcal{N}(x)\\ x' \neq x}} e(\{ x, x' \}) \geq -\alpha \# \mathcal{M}.
\end{align*}
\end{lemma}
\begin{proof}
If $x \in X$, $x' \notin \mathcal{N}(x)$, and $x' \neq x$, then, by the definition of $\mathcal{N}(x)$, $|y(x) - y(x')| > 1 + \alpha$. If also $|y(x) - y(x')| < \frac{3}{2}$, then $e(\{ x, x' \}) \geq -\alpha$ by assumption~\eqref{eq:V2_mid_range}, otherwise $e(\{ x, x' \}) = 0$ by assumption~\eqref{eq:V2_long_range}. Using these estimates on $e(\{ x, x' \})$, we have
\begin{align*}
	\frac{1}{2} \sum_{x \in X} \sum_{\substack{x' \notin \mathcal{N}(x)\\ x' \neq x}} e(\{ x, x' \})
	&= \frac{1}{2} \sum_{x \in X} \sum_{\substack{x' \in X\\ 1 + \alpha < |y(x') - y(x)| < \frac{3}{2}}} e(\{ x, x' \})\\
	&\geq -\frac{\alpha}{2} \sum_{x \in X} \sum_{\substack{x' \in X\\ 1 + \alpha < |y(x') - y(x)| < \frac{3}{2}}}  1\\
	&= -\alpha \# \mathcal{M},
\end{align*}
where the last equality holds because the sum double-counts pairs in $\mathcal{M}$. This proves the claim.
\end{proof}

The next result states that in the ground state a regular atom has no mid-range interactions. The proof has some similarities to the proof of Lemma~3.1 in \cite{mainini14}, which gives an upper bound on the number of neighbors a particle has in a ground state configuration.
\begin{lemma}
\label{lem:rigidity}
There exists $\alpha_0 \in \left( 0, \frac{1}{4} \right)$ and $\epsilon > 0$, such that
\begin{itemize}
	\item for all $\alpha \in (0, \alpha_0)$,
	\item for all potentials $V$ of form~\eqref{eq:energy_full}, with $V_2$ satisfying assumptions \eqref{eq:V2_smooth} -- \eqref{eq:V2_normalized} and $V_3$ satisfying assumptions \eqref{eq:V3_smooth} -- \eqref{eq:V3_angle72},
	\item for all ground states $y : X \rightarrow \mathbb{R}^2$ of $V(\cdot)$, and
	\item for all $x \in G_\epsilon, x' \in X$ with $x' \neq x$,
\end{itemize}
we have
\begin{itemize}
	\item $|y(x') - y(x)| \leq 1 + \alpha$, or
	\item $|y(x') - y(x)| \geq \frac{3}{2}$.
\end{itemize}
\end{lemma}

\begin{proof}
Let $\alpha_0 > 0$ be such that the inequality~\eqref{eq:min_dist} in Lemma~\ref{lem:min_dist} holds. Let $V$ satisfy the given assumptions, and let $y : X \rightarrow \mathbb{R}^2$ be a ground state. Let $x \in G_\epsilon$. To show that all particles are either closer than $1 + \alpha$ or further than $\frac{3}{2}$ away from $x$, we will use a proof by contradiction. Suppose that $p$ is such that $1 + \alpha < |y(p) - y(x)| < \frac{3}{2}$. We first show that $p$ is part of a triple $(x, p, p')$ with a bond angle less than $\frac{2\pi}{5}$. We then show that this contradicts the fact that the configuration is a ground state.

Label $\mathcal{N}(x)$ as $x_1$, $x_2$, and $x_3$ and label the bond angles as $\theta_1$, $\theta_2$, and $\theta_3$, as described in Section~\ref{sec:nbhd_energy_est}. Since $x \in G_{\epsilon}$, we have $\theta_i \in \left( \frac{2\pi}{3} - \epsilon, \frac{2\pi}{3} + \epsilon \right)$. If $\alpha_0 < \frac{1}{4}$, then we have that $p$ can not be along the same ray as any of the $x_i$. Otherwise, by~\eqref{eq:min_dist}, it would be further than $2 - 2\alpha > \frac{3}{2}$ away from $x$. So, $p$ and the $x_i$ can be placed in a counterclockwise ordering around $x$ beginning from the $e_1$-direction. Suppose without loss of generality that in this ordering $p$ follows $x_1$ and precedes $x_2$. See Figure~\ref{fig:mid_range}. Then we have $\theta_{x_1, x, p} + \theta_{p, x, x_2} + \theta_2 + \theta_3 = 2\pi$. Let $\epsilon = \frac{\pi}{15}$. Then $\theta_2, \theta_3 > \frac{2\pi}{3} - \frac{\pi}{15}$. Thus, $\theta_2 + \theta_3 > \frac{4\pi}{3} - \frac{2\pi}{15} = \frac{6\pi}{5}$. Therefore,
\begin{equation*}
	\theta_{x_1, x, p} + \theta_{p, x, x_2} = 2\pi - (\theta_2 + \theta_3) < 2\pi - \frac{6\pi}{5} = \frac{4\pi}{5}.
\end{equation*}
Thus, it is impossible that both $\theta_{x_1, x, p} \geq \frac{2\pi}{5}$ and $\theta_{p, x, x_2} \geq \frac{2\pi}{5}$. So, there exists a triple $(x, p, p')$ with $p' \in \{ x_1, x_2 \} \subset \mathcal{N}(x)$ and bond angle $\theta_{p, x, p'} \leq \frac{2\pi}{5}$.

\begin{figure}[tb]
\centering
\includegraphics[width=0.30\textwidth]{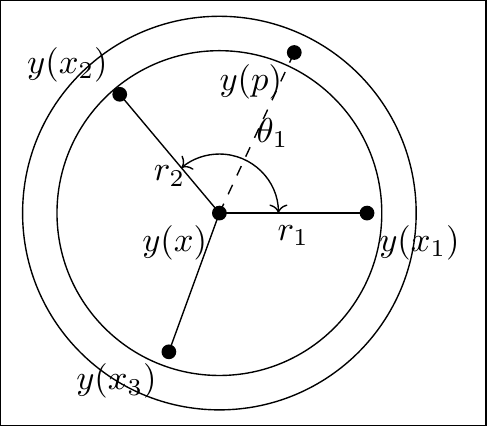}
\caption{Figure depicting particle $x$, and the particles in $B \left( y(x), \frac{3}{2} \right)$. The circles show the boundary of the annulus $B \left( y(x), \frac{3}{2} \right) \setminus B \left( y(x), 1 + \alpha \right)$.}
\label{fig:mid_range}
\end{figure}

The change in energy upon moving the particle $x$ to infinity must be non-negative since we are in a ground state. This change in energy is
\begin{equation}
\label{eq:energyChange}
	-\sum_{\substack{x' \in X\\ x' \neq x}} e(\{x, x'\}) - \frac{1}{2} \sum_{\substack{x', x'' \in X\\ x, x', x'' \text{ distinct}}} a(x, x', x'') - \sum_{\substack{x', x'' \in X\\ x, x', x'' \text{ distinct}}} a(x', x, x'') \geq 0.
\end{equation}

Since $x \in G_{\epsilon}$, it has three neighbors closer than $1 + \alpha$, which have edge energy $-1$ or greater. All other particles $x'$ such that $e(\{x, x'\})$ is non-zero are in the annulus centered at $x$ with inner radius $1+\alpha$ and outer radius $\frac{3}{2}$. Since $\alpha > 0$, this annulus is contained in another with inner radius $1$ and outer radius $\frac{3}{2}$. We can infer from the minimum distance property~\eqref{eq:min_dist} that there exists $C >0$, independent of $\alpha$, such that there can be at most $C$ atoms in this annulus, by covering this region with balls of radius $\frac{2}{3}$ as in Lemma~\ref{lem:min_dist}. Thus, using assumptions~\eqref{eq:V2_mid_range} and~\eqref{eq:V2_normalized} on the potential $V_2$,
\begin{equation}
\label{eq:rigidity_edge}
	\sum_{\substack{x' \in X\\ x'\neq x}} e(\{x, x'\}) \geq -3 - C\alpha.
\end{equation}

Also, since the angle energy is non-negative, we have
\begin{equation}
\label{eq:rigidity_angle1}
	\sum_{\substack{x', x'' \in X\\ x, x', x'' \text{ distinct}}} a(x', x, x'') \geq 0
\end{equation}
and because there is at least one bond angle $\theta_{p, x, p'} \leq \frac{2\pi}{5}$ with $|y(p) - y(x)| < \frac{3}{2}$ and $|y(p') - y(x)| < \frac{3}{2}$, by assumption~\eqref{eq:V3_angle72}, we have
\begin{equation}
\label{eq:rigidity_angle2}
	\frac{1}{2} \sum_{\substack{x', x'' \in X\\ x, x', x'' \text{ distinct}}} a(x, x', x'') \geq \frac{1}{2} \left( a(x, p, p') + a(x, p', p) \right) = V_3(\theta_{p, x, p'}) \geq V_3 \left( \frac{2\pi}{5} \right) \geq 4.
\end{equation}

Therefore, substituting \eqref{eq:rigidity_edge}, \eqref{eq:rigidity_angle1}, and \eqref{eq:rigidity_angle2} into \eqref{eq:energyChange}, we have
\begin{equation*}
	0 \leq \left( 3 + C\alpha \right) + 0 - 4 = -1 + C\alpha.
\end{equation*}

Choosing $\alpha_0$ to be the lesser of the previous value and $\frac{1}{2C}$, we have for $0 < \alpha < \alpha_0$ that $-\frac{1}{2} > 0$, which is a contradiction. Thus, for a ground state, there can be no such point $p$ with $1 + \alpha \leq |y(p) - y(x)| \leq \frac{3}{2}$. This proves the claim.
\end{proof}

As a result, we get the following upper bound on the number of mid-range interactions.
\begin{lemma}
\label{lem:ub_midrange}
There exists $\alpha_0, \epsilon, C > 0$ such that
\begin{itemize}
	\item for all $\alpha$ with $0 < \alpha < \alpha_0$,
	\item for all potentials $V$ of form~\eqref{eq:energy_full}, with $V_2$ satisfying assumptions \eqref{eq:V2_smooth} -- \eqref{eq:V2_normalized} and $V_3$ satisfying assumptions \eqref{eq:V3_smooth} -- \eqref{eq:V3_angle72}, and
	\item for all ground states $y : X \rightarrow \mathbb{R}^2$ of $V(\cdot)$,
\end{itemize}
we have
\begin{equation*}
	\# \mathcal{M} \leq C (N - \# G_{\epsilon})
\end{equation*}
\end{lemma}

\begin{proof}
Choose $\alpha_0, \epsilon > 0$ such that the hypotheses of Lemma~\ref{lem:rigidity} hold. We first claim that
\begin{equation*}
	\mathcal{M} = \mathcal{M}_d := \left\{ \{x, x'\} \subset G_{\epsilon}^{c} : 1+\alpha < |y(x) - y(x')| < \frac{3}{2} \right\},
\end{equation*}
where $\mathcal{M}_d$ are the mid-range pairs of defected atoms. The inclusion $\mathcal{M}_d \subseteq \mathcal{M}$ follows directly from the definition of $\mathcal{M}$. The opposite direction can be proved using the previous lemma. Suppose that $\{ x, x' \} \in \mathcal{M}$, i.e.\ $x, x' \in X$ with $1 + \alpha \leq |y(x') - y(x)| \leq \frac{3}{2}$. Then, by the contrapositive of Lemma~\ref{lem:rigidity}, $x, x' \in G_{\epsilon}^c$. Thus, $\mathcal{M} \subseteq \mathcal{M}_d$.

Therefore, we can enumerate the mid-range interactions by using a sum to count all ordered pairs where both atoms are defected and the second atom is in the annulus of inner radius $1+\alpha$ and outer radius $\frac{3}{2}$ centered at the first atom. A factor of $\frac{1}{2}$ will be used to normalize for the double counting:
\begin{align*}
	\# \mathcal{M}
	&= \frac{1}{2} \sum_{x \in G_{\epsilon}^c} \sum_{\substack{x' \in G_{\epsilon}^c\\ 1+\alpha < |y(x')-y(x)| < \frac{3}{2}}} 1.
\end{align*}

As in the proof of Lemma~\ref{lem:rigidity}, by the minimum distance property~\eqref{eq:min_dist}, there exists a constant $C > 0$, independent of $\alpha$, such that there at most $2 C$ particles in an annulus of inner radius $1+\alpha$ and outer radius $\frac{3}{2}$. Thus, we have
\begin{align*}
	\# \mathcal{M}
	& \leq \frac{1}{2} \sum_{x \in G_{\epsilon}^c} 2C\\
	& = C \# G_{\epsilon}^c\\
	& = C (N - \# G_{\epsilon}).
\end{align*}
This proves the claim.
\end{proof}

\subsection{Final estimate on the potential}
We now gather the results of the previous sections to prove our main theorem (Theorem~\ref{thm:main}), first obtaining the estimate from Theorem~\ref{thm:main_ineq}.

\begin{proof}[Proof of Theorems~\ref{thm:main} and~\ref{thm:main_ineq}]
Let the potential $V$ satisfy the given assumptions. Let $\{ y \}$ be a ground state configuration. Let $\alpha_0$ be such that the hypotheses of Lemma~\ref{lem:min_dist} are satisfied. Then, the configuration $\{y\}$ satisfies the minimum distance inequality~\eqref{eq:min_dist}. If necessary, reduce $\alpha_0$ such that the hypotheses of Lemmas~\ref{lem:lb_all} and~\ref{lem:lb_defected} hold.

We begin with the estimate \eqref{eq:energy_decomp} from Lemma~\ref{lem:energy_decomp}:
\begin{align}
	V &\geq
	\sum_{x \in G_\epsilon} V_{\mathcal{N}} (x)
	+ \sum_{x \notin G_\epsilon} V_{\mathcal{N}} (x)
	+ \frac{1}{2} \sum_{x \in X}\ \sum_{\substack{x' \notin \mathcal{N}(x)\\ x' \neq x}} e(\{ x, x' \}).
\end{align}

Using Lemma~\ref{lem:mid-range} to estimate the sum over non-neighbor pairs, we have
\begin{align*}
	V &\geq
	\sum_{x \in G_\epsilon} V_{\mathcal{N}} (x)
	+ \sum_{x \notin G_\epsilon} V_{\mathcal{N}} (x)
	- \alpha \# \mathcal{M}.
\end{align*}

Now, using Lemma~\ref{lem:lb_all} to estimate the sum over the regular atoms, this becomes
\begin{align*}
	V &\geq
	\sum_{x \in G_\epsilon} \left( -\frac{3}{2} + 3 V_3(2\pi/3) + W_{e}(x) + W_{a}(x) \right)
	+ \sum_{x \notin G_\epsilon} V_{\mathcal{N}} (x)
	- \alpha \# \mathcal{M}.
\end{align*}

By Lemma~\ref{lem:lb_defected}, we can estimate the sum over defected atoms. Thus, there exists $\Delta_{\epsilon} > 0$ such that
\begin{align*}
	V &\geq
	\sum_{x \in G_\epsilon} \left( -\frac{3}{2} + 3 V_3(2\pi/3) + W_{e}(x) + W_{a}(x) \right)
	+ \sum_{x \notin G_\epsilon} \left( -\frac{3}{2} + 3 V_3(2\pi/3) + \Delta_\epsilon \right)
	- \alpha \# \mathcal{M}.
\end{align*}
	
Adding up the constant parts of the sums and simplifying, we have
\begin{align*}
	\begin{split}
	V &\geq
	\# G_{\epsilon} \left( -\frac{3}{2} + 3 V_3(2\pi/3) \right) + \sum_{x \in G_\epsilon} \left( W_{e}(x) + W_{a}(x) \right)\\
		& \qquad \qquad + (N - \# G_{\epsilon}) \left( -\frac{3}{2} + 3 V_3(2\pi/3) + \Delta_\epsilon \right) - \alpha \# \mathcal{M}
	\end{split}\\
	&= \left( -\frac{3}{2} + 3 V_3(2\pi/3) \right) N + \Delta_\epsilon (N - \# G_{\epsilon}) + \sum_{x \in G_\epsilon} \left( W_{e}(x) + W_{a}(x) \right) - \alpha \# \mathcal{M}.
\end{align*}

Using the estimate on the number of mid-range pairs from Lemma~\ref{lem:ub_midrange}, there exists $\epsilon > 0, C > 0$ such that this becomes
\begin{align*}
	V &\geq
	\left( -\frac{3}{2} + 3 V_3(2\pi/3) \right) N + \Delta_\epsilon (N - \# G_{\epsilon}) + \sum_{x \in G_\epsilon} \left( W_{e}(x) + W_{a}(x) \right) - C \alpha (N - \# G_{\epsilon}).
\end{align*}
If necessary, reduce $\alpha_0$ to $\alpha_0 = \frac{1}{2C}\Delta_\epsilon$ so that $C \alpha < \frac{1}{2}\Delta_\epsilon$. Then, we have
\begin{align}
	\label{eq:main_ineq}
	V &\geq
	\left( -\frac{3}{2} + 3 V_3(2\pi/3) \right) N + \frac{\Delta_\epsilon}{2} (N - \# G_{\epsilon}) + \sum_{x \in G_\epsilon} \left( \sum_{j=1}^3 \frac{1}{2} \left| r_j - 1 \right|^2 + \sum_{j=1}^3 \frac{m}{2} \left| \theta_j - \frac{2\pi}{3} \right|^2 \right)\\
	\label{eq:energy_lb}
	&\geq \left( -\frac{3}{2} + 3 V_3(2\pi/3) \right) N.
\end{align}
The inequality~\eqref{eq:main_ineq} is the claim of Theorem~\ref{thm:main_ineq}, with $G = G_{\epsilon}$. Along with the upper bound~\eqref{eq:energy_ub}, the lower bound~\eqref{eq:energy_lb} proves Theorem~\ref{thm:main}.
\end{proof}

\subsection{Estimates on the number of defected atoms and the excess surface energy}
We next obtain a lower bound on the number of defected atoms by showing that it is not possible for all particles to have three bonds with bond angles $2 \pi/3$; there must be some particles on the boundary that contribute surface energy. The number of such particles must grow at least as fast as $N^{1/2}$.

First we prove a geometric covering result.

\begin{lemma}
\label{lem:covering}
	For any $\theta_{\max} \in (0, \pi)$ and $\alpha \in \left( 0, \frac{1}{3} \right)$, there exists $R_b > 4/3$ such that for any configuration $y : X \rightarrow \mathbb{R}^2$ which satisfies \eqref{eq:min_dist}, we have the following implication for a particle $x \in X$ and its neighborhood $\mathcal{N}(x) =: \{ x_i \}$:

If $x \in X$ has three or more neighbors $\{x_i\}$ with consecutive bond angles $\theta_i := \theta_{x_i, x, x_{i+1}} < \theta_{\max} \text{ for } i = 1, \ldots, M$, then the balls of radius $R_b$ centered at the $\{y(x_i)\}$ will completely cover the ball of radius $R_b$ centered at $y(x)$, i.e.
	\begin{equation*}
		B(y(x), R_b) \subset \bigcup_{i=1}^M B(y(x_i), R_b),
	\end{equation*}
\end{lemma}

\begin{proof}
Let $\theta_{\max} \in (0, \pi)$ and $\alpha \in  \left( 0, \frac{1}{3} \right)$. Choose $R_b > 4/3$ such that $\theta_{max} < 2 \cos^{-1} \left( \frac{2}{3 R_b} \right)$. This is possible because $\theta_{max} < \pi$, $h(R) := 2 \cos^{-1} \left( \frac{2}{3R} \right)$ is an increasing function, and $\lim_{R \rightarrow \infty} h(R) = \pi$.

Now, consider a configuration $y : X \rightarrow \mathbb{R}^2$ which satisfies \eqref{eq:min_dist}. Suppose the particle $x$ has three or more neighbors $\{x_i\}$ with consecutive bond angles $\theta_i < \theta_{\max}$. We wish to show that the balls of radius $R_b$ centered at $\{y(x_i)\}$ will completely cover the ball of radius $R_b$ centered at $y(x)$. We consider the different sectors of the circle corresponding to each bond angle $\theta_i$. First, consider the sector between the vectors $y(x_1) - y(x)$ and $y(x_2) - y(x)$, corresponding to angle $\theta_1$. See the schematic in Figure~\ref{fig:hex_nbhd}, where $x$ has three neighbors. Consider the point $p$ that is distance $R_b$ from both $y(x_1)$ and $y(x_2)$ and which lies in the sector defined by angle $\theta_1$. Denote the distance from $y(x)$ to $p$ by $D$. Since $r_i < 4/3 < R_b$, the sector will be covered as long as $D > R_b$. We decompose $\theta_1$ as the sum of two angles. Define $\beta$ as the angle between the vectors $y(x_1) - y(x)$ and $p - y(x)$ and define $\gamma$ as the angle between the vectors $p - y(x)$ and $y(x_2) - y(x)$. Then, $\theta_1 = \beta + \gamma$. Define $\beta_0 := \cos^{-1} \left( \frac{r_1}{2 R_b} \right)$ and $\gamma_0 := \cos^{-1} \left( \frac{r_2}{2 R_b} \right)$. By geometric reasoning, we have:
\begin{itemize}
	\item $D = R_b \iff \beta = \beta_0 \iff \gamma = \gamma_0$,
	\item $D < R_b \iff \beta > \beta_0 \iff \gamma > \gamma_0$, and
	\item $D > R_b \iff \beta < \beta_0 \iff \gamma < \gamma_0$. 
\end{itemize}

\begin{figure}[tb]
\centering
\includegraphics[scale=1.0]{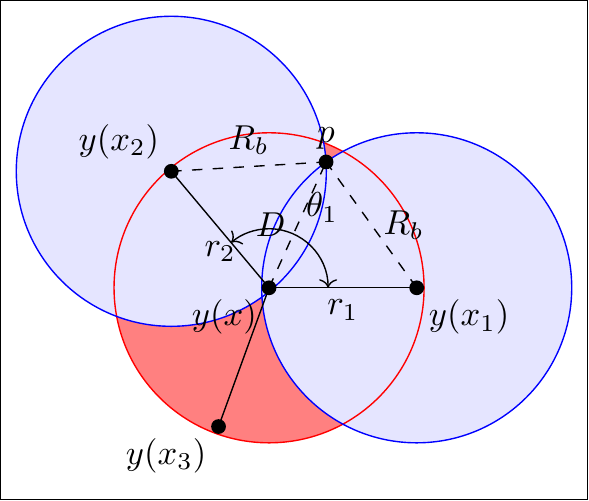}
\includegraphics[scale=1.0]{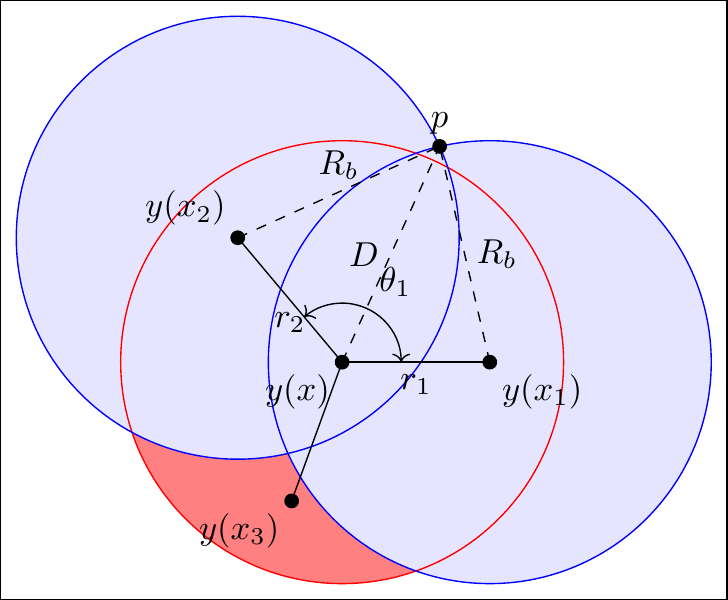}
\caption{Figures depicting the balls $B(y(x), R_b)$, $B(y(x_1), R_b)$, and $B(y(x_2), R_b)$. Left: a case where $R_b$ is such that $D < R_b$. Right: a case where $R_b$ is such that $D > R_b$.}
\label{fig:hex_nbhd}
\end{figure}

We want to show that our choice of $R_b$ guarantees that $D > R_b$. We have that $r_i := {|y(x_i)-y(x)|} \leq 1 + \alpha < 4/3$, so $\frac{r_i}{2 R_b} < \frac{2}{3R_b}$. Since $\cos^{-1}(c)$ is a decreasing function of $c$, this implies that
\begin{align*}
	\beta_0 &= \cos^{-1} \left( \frac{r_1}{2 R_b} \right) > \cos^{-1} \left( \frac{2}{3 R_b} \right) \text{, and}\\
	\gamma_0 &= \cos^{-1} \left( \frac{r_2}{2 R_b} \right) > \cos^{-1} \left( \frac{2}{3 R_b} \right).
\end{align*}
Therefore,
\begin{equation*}
	\beta + \gamma = \theta_1 < \theta_{max} < 2 \cos^{-1} \left( \frac{2}{3 R_b} \right) < \beta_0 + \gamma_0.
\end{equation*}

Now, since $\beta + \gamma < \beta_0 + \gamma_0$, we must have either $\beta < \beta_0$ or $\gamma < \gamma_0$. By the reasoning above, this implies $D > R_b$. Therefore the sector of the circle $B(y(x), R_b)$ between the directions $y(x_1) - y(x)$ and $y(x_2) - y(x)$ is covered by the balls $B(y(x_1), R_b)$ and $B(y(x_2), R_b)$. Since $\theta_i \leq \theta_{max}$ for all $i$, we can show that the other sectors are covered for the same $R_b$, using the same argument. Thus, the entire circle will be covered for this choice of $R_b$.
\end{proof}

\begin{proposition}
	For all $\epsilon \in (0, \pi/3)$, there exists $R_{\epsilon} > 4/3$ such that for any configuration $y : X_N \rightarrow \mathbb{R}^2$ which satisfies \eqref{eq:min_dist} with $\alpha \in \left( 0, \frac{1}{3} \right)$, we have
	\begin{equation*}
    		\frac{1}{3 R_{\epsilon}} N^{1/2} \leq N - \# G_{\epsilon}.
    	\end{equation*}
\end{proposition}

\begin{proof}
    Define $\theta_{\max} := 2\pi/3 + \epsilon < \pi$. Choose $R_{\epsilon} > 4/3$ as in Lemma~\ref{lem:covering}.
    
    Let $y : X_N \rightarrow \mathbb{R}^2$ be a configuration which satisfies \eqref{eq:min_dist} with $\alpha \in (0, \frac{1}{3})$. Consider the set $S := \cup_{x \in X_N} B(y(x), R_{\epsilon})$. We have $\cup_{x \in X_N} B(y(x), 1/3) \subset S$, and since the particles are at least distance $1 - \alpha$ apart, by \eqref{eq:min_dist}, and $\alpha < 1/3$, this set is also a disjoint union. Therefore,
    \begin{equation}
    \label{eq:area_S}
    	\mathrm{Area}(S) \geq N \pi/9.
    \end{equation}
    
    Now, consider a particle $x \in G_{\epsilon}$ and its neighbors $\mathcal{N}(x) = \{ x_i \}_{i=1}^3$. By the definition of $G_{\epsilon}$ and $\mathcal{N}(x)$, $\max_i \theta_i < 2\pi/3 + \epsilon = \theta_{\max} < \pi$. Then, by Lemma~\ref{lem:covering}, for our choice of $R_{\epsilon}$, the balls of radius $R_{\epsilon}$ centered at the three neighbors of $x$ will completely cover the ball of radius $R_{\epsilon}$ centered at $y(x)$. Thus, the ball around $x$ does not contribute to the perimeter of $S$.  Since $x \in G_{\epsilon}$ was arbitrary, none of the balls around the particles in $G_{\epsilon}$ contribute to the perimeter of $S$. Therefore,
    \begin{equation*}
    	\mathrm{Per}(S) \leq  \sum_{x \notin G_{\epsilon}} \mathrm{Per} (B(y(x), R_{\epsilon})) = (N - \# G_{\epsilon}) 2\pi R_{\epsilon}.
    \end{equation*}
    
    By \eqref{eq:area_S} and the isoperimetric inequality for the plane, we have
    \begin{equation*}
    	4 N \pi^2/9 \leq 4\pi  \mathrm{Area}(S) \leq \left( \mathrm{Per}(S) \right)^2 \leq  (N - \# G_{\epsilon})^2 4 \pi^2 R_{\epsilon}^2,
    \end{equation*}
    or 
    \begin{equation*}
    	\frac{1}{3R_{\epsilon}} N^{1/2} \leq N - \# G_{\epsilon},
    \end{equation*}
    i.e. that the number of defected atoms is bounded below by a constant times $N^{1/2}$.
\end{proof}

This can be combined with our main estimate \eqref{eq:main_ineq} to yield
\begin{align}
	V
	&\geq \left( -\frac{3}{2} + 3 V_3(2\pi/3) \right) N + \frac{\Delta_\epsilon}{6R_{\epsilon}} N^{1/2} + \sum_{x \in G_\epsilon} \left( \sum_{j=1}^3 \frac{1}{2} \left| r_j - 1 \right|^2 + \sum_{j=1}^3 \frac{m}{2} \left| \theta_j - \frac{2\pi}{3} \right|^2 \right)\\
	&\geq \left( -\frac{3}{2} + 3 V_3(2\pi/3) \right) N + \frac{\Delta_\epsilon}{6R_{\epsilon}} N^{1/2}.
\end{align}

We can prove a corresponding upper bound on the ground state energy by considering trial configurations.

\begin{proposition}
\label{prop:ub_Vmin}
	For any potential of the form~\eqref{eq:energy_full}, with $V_2$ satisfying assumptions \eqref{eq:V2_smooth} -- \eqref{eq:V2_normalized} and $V_3$ satisfying assumptions \eqref{eq:V3_smooth} -- \eqref{eq:V3_angle72}, the following equation holds:
\begin{equation*}
	\min_{ y: X_N \rightarrow \mathbb{R}^2} V(\{y\}) \leq \left( -\frac{3}{2} + 3 V_3 \left( 2\pi / 3 \right) \right) N + \sqrt{\frac{3}{2}} N^{1/2} + 1.
\end{equation*}
\end{proposition}

\begin{proof}
We obtain the upper bound on the ground state energy by considering the trial configurations from the work of Mainini and Stefanelli \cite{mainini14}. For $N = 6k^2$, $k$ an integer, these configurations are highly symmetric states which Mainini and Stefanelli call ``daisies.''  For other values of $N$, they are a geometric interpolation between two daisy configurations. All of these configurations are subsets of the honeycomb lattice $H$. As a result, only first-neighbor pairs and first-neighbor triples contribute to the energy. To clarify how these pairs and triples are to be counted, we note that
\begin{align}
	\text{\# of first-neighbor pairs} &= \frac{1}{2} \sum_{x \in X_N} \# \mathcal{N}(x) \text{, and}\\
	\text{\# of (non-equivalent) first-neighbor triples} &= \frac{1}{2} \sum_{x \in X_N} (\# \mathcal{N}(x)) (\# \mathcal{N}(x) - 1).
\end{align}

Denote $\left \lfloor w \right \rfloor := \max \{ n \in \mathbb{Z} : n \leq w \}$. The configurations $\{y_N\}$ constructed by Mainini and Stefanelli in \cite[Proposition 5.1]{mainini14} satisfy the following estimate:
\begin{equation}
\label{eq:lb_n_pairs2}
	\text{\# of first-neighbor pairs} \geq \left \lfloor \frac{3}{2} N - \sqrt{\frac{3}{2}} N^{1/2} \right \rfloor
	\geq \frac{3}{2} N - \sqrt{\frac{3}{2}} N^{1/2} - 1.
\end{equation}

For our energy, since the configurations $\{y_N\}$ are a subset of the honeycomb lattice, we have
\begin{align*}
	V(\{ y_N \})
	&= (-1) \left( \text{\# of first-neighbor pairs} \right)
	+ V_3(2\pi/3) \left( \text{\# of first-neighbor triples} \right).
\end{align*}
Since each atom has at most 3 first-neighbors, the number of first-neighbor triples is less than $3N$. Combining this with~\eqref{eq:lb_n_pairs2}, we have
\begin{align*}
	V (\{y_N\} )
	& \leq -\left( \frac{3}{2} N - \sqrt{\frac{3}{2}} N^{1/2} -1 \right) + 3 V_3 \left(2\pi / 3 \right) N\\
	& = \left( -\frac{3}{2} + 3 V_3 \left( 2\pi / 3 \right) \right) N + \sqrt{\frac{3}{2}} N^{1/2} + 1.
\end{align*}
Since the ground state energy must be less than our equal to the energy of this trial configuration, this proves the claim.
\end{proof}

Combining the above result with the inequality~\eqref{eq:main_ineq}, neglecting the elastic terms, we have for a ground state configuration $y_{\min}: X_N \rightarrow \mathbb{R}^2$:
\begin{equation}
\begin{split}
	\left( -\frac{3}{2} + 3 V_3(2\pi/3) \right) N + \frac{\Delta_\epsilon}{2} (N - \# G_{\epsilon})
	&\leq V \left(\{ y_{\min} \} \right)\\
	&\leq \left( -\frac{3}{2} + 3 V_3 \left( 2\pi / 3 \right) \right) N + \sqrt{\frac{3}{2}} N^{1/2} + 1.
\end{split}	
\end{equation}

Subtracting the first term from the right and left of this inequality and simplifying, we have an upper bound for the number of defected atoms in a ground state configuration:
\begin{align}
	N - \# G_{\epsilon}
	\leq \frac{2}{\Delta_\epsilon} \left( \sqrt{\frac{3}{2}} N^{1/2} + 1 \right) = \frac{\sqrt{6}}{\Delta_{\epsilon}} N^{1/2} + \frac{2}{\Delta_\epsilon}.
\end{align}

\section{Formation of a honeycomb lattice}
\label{sec:honeycomb}
Subject to periodic boundary conditions, the minimizer of the energy is a honeycomb lattice. This can be derived from a new version of estimate \eqref{eq:main_ineq} in a similar fashion as in \cite{theil06, e09}.

\subsection{Definitions and theorem}

First, we define what is meant by ``periodic boundary conditions.'' This entails an infinite number of particles, so we re-define the energy to include only contributions of particles contained in a ``reference cell,'' which is repeated periodically to form the configuration. Let $L \in \mathbb{N}$. Define $U = \frac{1}{2} \begin{pmatrix} 2\sqrt{3} & \sqrt{3}\\ 0 & 3 \end{pmatrix} Q$, where $Q = [0,1) \times [0, 1) \subset \mathbb{R}^2$ is the semi-open unit cell.

The energy of the configuration $y : H \rightarrow \mathbb{R}^2$ is defined as
\begin{align}
\label{eq:energy_periodic}
	V_L^{per}( \{y\} ) = \frac{1}{2} \sum_{x \in H \cap LU} \sum_{\substack{x' \in H\\ x' \neq x}} e( \{x, x'\} ) + \frac{1}{2} \sum_{x \in H \cap LU} \sum_{\substack{x', x'' \in H\\ x, x', x'' \text{ distinct}}} a(x, x', x'').
\end{align}

Let
$A_2 = \frac{1}{2} \begin{pmatrix} 2 & 1\\ 0 & \sqrt{3} \end{pmatrix} \mathbb{Z}^2 \subset \mathbb{R}^2$.
Thus, $\sqrt{3} A_2$ is the triangular lattice generated by the vectors $a_1$ and $a_2$ used to define the honeycomb lattice. The constraint is
\begin{equation*}
	y \in Y_L^{per} := \{ y : H \rightarrow \mathbb{R}^2 \ | \ y(x) - y(x') = x - x' \text{ if } x - x' \in L \sqrt{3} A_2 \}.
\end{equation*}

Using these definitions, we can show in the following theorem that the ground state is a translated honeycomb lattice.
\begin{theorem}
\label{thm:periodic}
There exists a constant $\alpha_0 \in (0, \frac{1}{3})$ such that for any $0 < \alpha < \alpha_0$, $L \in \mathbb{N}$, any potential $V_L^{per}$ of the form~\eqref{eq:energy_periodic}, with $V_2$ satisfying assumptions \eqref{eq:V2_smooth} -- \eqref{eq:V2_normalized} and $V_3$ satisfying assumptions \eqref{eq:V3_smooth} -- \eqref{eq:V3_angle72}, and any ground state $y_{\min} : H \rightarrow \mathbb{R}^2$ of $V_L^{per} (\cdot)$ subject to $y \in Y_L^{per}$, there exists a translation vector $\tau \in \mathbb{R}^2$ such that
\begin{equation*}
	\{y_{\min} (x) + \tau \ : \ x \in H\} = H.
\end{equation*}
\end{theorem}

The proof is based on a modified version of the inequality~\eqref{eq:main_ineq}. This inequality required Lemma~\ref{lem:min_dist}, a minimum inter-particle distance result, and Lemma~\ref{lem:rigidity}, a result stating that regular atoms cannot have mid-range interactions. When we proved these results earlier in this article, we considered the change in energy when a particle or set of particles was moved to infinity such that their mutual interactions went to zero. Since the configurations are infinite in the periodic case, this is no longer an option. Now, we need to establish these results when a particle and its periodic images are removed. We need to re-define the energy to allow for the removal of $L$-periodic sets.

Let $L \in \mathbb{N}$. A set $X \subset H$ is called $L$-periodic if $X + L \sqrt{3} A_2 = X$. We introduce an equivalence relation $\sim$ on subsets of an $L$-periodic set $X$ such that: two subsets $\omega, \omega' \subset X$ satisfy $\omega \sim \omega'$ if there is a vector $\tau \in L \sqrt{3} A_2$ for which $\omega' = \omega + \tau$. We say a map $y: X \rightarrow \mathbb{R}^2$ is $L$-periodic if $y(x + \tau) = y(x) + \tau$ for all $x \in X$, $\tau \in L \sqrt{3} A_2$. (This is equivalent to the previous definition of $Y_L^{per}$.) For an $L$-periodic map, the set of regular atoms $G_{\epsilon}$ and the set of midrange pairs $\mathcal{M}$ are periodic sets, and we can define the quotient sets $\widetilde{X} := X / \sim$, $\widetilde{G}_{\epsilon} := G_{\epsilon} / \sim$, and $\widetilde{\mathcal{M}} := \mathcal{M} / \sim$.

The new energy, defined for an $L$-periodic set $X \subset H$ is
\begin{align}
\label{eq:energy_periodic_subset}
	V_L^{per}(X, \{y\} ) = \frac{1}{2} \sum_{x \in X \cap LU} \sum_{\substack{x' \in X\\ x' \neq x}} e( \{x, x' \} ) + \frac{1}{2} \sum_{x \in X \cap LU} \sum_{\substack{x', x'' \in X\\ x, x', x'' \text{ distinct}}} a(x, x', x'').
\end{align}

As there are $2L^2$ particles in the reference cell $H \cap LU$, there are only $2^{2L^2}$ possible $L$-periodic sets $X$. Thus, a minimizer $(X_{\min}, y_{\min})$ of $V_L^{per}$ exists.

\subsection{Minimum distance result and non-existence of mid-range interactions for regular atoms}

We formulate versions of Lemmas~\ref{lem:min_dist} and~\ref{lem:rigidity} which apply to the energy $V_L^{per} (\cdot, \cdot)$. Before we state and prove these results rigorously, we describe how the change in the pair potential upon removing periodic subsets differs from the non-periodic case. For $\mathcal{F}, \mathcal{G} \subset X$, we shall use the notation $\Pi \mathcal{F} := \mathcal{F} + L \sqrt{3} A_2$ for the periodization of the set $\mathcal{F}$,\footnote{For simplicity, for a singleton $\{ x \}$, we will denote $\Pi \{ x \}$ by $\Pi x$.} and
\begin{equation*}
	e(\mathcal{F}, \mathcal{G}) := \sum_{x \in \mathcal{F}} \sum_{x' \in \mathcal{G} \setminus \{ x \}} e( \{ x, x' \} )
\end{equation*}
for the pair energy where the first sum runs over indices $\mathcal{F}$ and the second sum runs over indices $\mathcal{G}$.

For the minimum distance result, in the non-periodic case, we decomposed the total pair potential as
\begin{align*}
	\frac{1}{2} e(X, X) &= \frac{1}{2} e(\mathcal{A}, X)
		+ \frac{1}{2} e(X \setminus \mathcal{A}, X)\\
		&= \frac{1}{2} e(\mathcal{A}, \mathcal{A})
		+ \frac{1}{2} e(\mathcal{A}, X \setminus \mathcal{A})
		+ \frac{1}{2} e(X \setminus \mathcal{A}, \mathcal{A})
		+ \frac{1}{2} e(X \setminus \mathcal{A}, X \setminus \mathcal{A})\\
	&= \frac{1}{2} e(\mathcal{A}, \mathcal{A})
		+ e(\mathcal{A}, X \setminus \mathcal{A})
		+ \frac{1}{2} e(X \setminus \mathcal{A}, X \setminus \mathcal{A})
\end{align*}
because the middle two terms in the first equation can be combined by symmetry. Now, in the periodic case, we will assume that $\mathcal{A} \subset X \cap LU$. Then, we will decompose $\frac{1}{2} e(X \cap LU, X)$ as
\begin{equation}
\begin{split}
\label{eq:pair_decomp_periodic}
	\frac{1}{2} e(X \cap LU, X) = {}& \frac{1}{2} e(\mathcal{A}, X) + \frac{1}{2} e((X \cap LU) \setminus \mathcal{A}, X)\\
		= {}& \frac{1}{2} e(\mathcal{A}, \mathcal{A})
		+ \frac{1}{2} e(\mathcal{A}, (\Pi \mathcal{A}) \setminus \mathcal{A})
		+ \frac{1}{2} e(\mathcal{A}, X \setminus (\Pi \mathcal{A})) \\
		& + \frac{1}{2} e((X \cap LU) \setminus \mathcal{A}, \Pi \mathcal{A}) \\
		& + \frac{1}{2} e((X \cap LU) \setminus \mathcal{A}, X \setminus (\Pi \mathcal{A})).
\end{split}
\end{equation}

We claim that $e((X \cap LU) \setminus \mathcal{A}, \Pi \mathcal{A})$ and $e(\mathcal{A}, X \setminus (\Pi \mathcal{A}))$ are equal. To see this, note that for all $x \in (X \cap LU) \setminus \mathcal{A}, x' \in \Pi \mathcal{A}$, there exists $\tau \in L\sqrt{3} A_2$ with $x' - \tau \in \mathcal{A}$ and $x - \tau \in X \setminus (\Pi \mathcal{A})$. Since $y$ is $L$-periodic, we have $e(\{x' - \tau, x - \tau\}) = e(\{x, x'\})$. Therefore, $e((X \cap LU) \setminus \mathcal{A}, \Pi \mathcal{A}) = e(\mathcal{A}, X \setminus (\Pi \mathcal{A}))$. As a result, \eqref{eq:pair_decomp_periodic} becomes
\begin{equation*}
\begin{split}
	\frac{1}{2} e(X \cap LU, X) = {}& \frac{1}{2} e(\mathcal{A}, \mathcal{A})
		+ \frac{1}{2} e(\mathcal{A}, (\Pi \mathcal{A}) \setminus \mathcal{A})
		+ e(\mathcal{A}, X \setminus (\Pi \mathcal{A})) \\
		& + \frac{1}{2} e((X \cap LU) \setminus \mathcal{A}, X \setminus (\Pi \mathcal{A})).
\end{split}
\end{equation*}

If we remove the $L$-periodic set $\Pi \mathcal{A}$ from $X$, then the total pair energy is
\begin{equation*}
	\frac{1}{2} e((X \setminus (\Pi \mathcal{A})) \cap LU, X \setminus (\Pi \mathcal{A}))
	= \frac{1}{2} e((X \cap LU) \setminus \mathcal{A}, X \setminus (\Pi \mathcal{A})).
\end{equation*}

Therefore, the change in the pair energy is
\begin{equation}
\label{eq:energy_change_A}
\begin{split}
	&\frac{1}{2} e((X \setminus (\Pi \mathcal{A})) \cap LU, X \setminus (\Pi \mathcal{A})) - \frac{1}{2} e(X \cap LU, X)\\
	& \quad = -\frac{1}{2} e(\mathcal{A}, \mathcal{A})
		- \frac{1}{2} e(\mathcal{A}, (\Pi \mathcal{A}) \setminus \mathcal{A})
		- e(\mathcal{A}, X \setminus (\Pi \mathcal{A})).
\end{split}
\end{equation}

For the proof of the non-existence of mid-range interactions for regular atoms, we consider similar decompositions where the set $\mathcal{A}$ is replaced by the single particle $x$. In the non-periodic case, we decomposed the pair potential as
\begin{align*}
	\frac{1}{2} e(X, X) = e(x, X \setminus \{x\})
		+ \frac{1}{2} e(X \setminus \{x\}, X \setminus \{x\}). 
\end{align*}
Now, in the periodic case, we decompose it as
\begin{align*}
	\frac{1}{2} e(X \cap LU, X) &= \frac{1}{2} e(x, X) + \frac{1}{2} e((X \cap LU) \setminus \{x\}, X)\\
		&= \frac{1}{2} e(x, (\Pi x) \setminus \{x\})
		+ e(x, X \setminus (\Pi x))\\
		& \quad + \frac{1}{2} e((X \cap LU) \setminus \{x\}, X \setminus (\Pi x)). 
\end{align*}
The change in energy upon removing the set $\Pi x$ is
\begin{equation}
\label{eq:energy_change_x}
\begin{split}
	\frac{1}{2} e((X \setminus (\Pi x)) \cap LU, X \setminus (\Pi x)) - \frac{1}{2} e(X \cap LU, X)
	= -\frac{1}{2} e(x, (\Pi x) \setminus \{x\})
		- e(x, X \setminus (\Pi x)).
\end{split}
\end{equation}

We now state and prove the periodic version of the minimum distance result.
\begin{lemma}
\label{lem:min_dist_periodic}
	There exists a constant $\alpha_0 \in \left( 0, \frac{1}{3} \right)$ such that for all $\alpha \in (0, \alpha_0)$, all $L \in \mathbb{N}$, and all potentials of form~\eqref{eq:energy_periodic_subset}, with $V_2$ satisfying assumptions \eqref{eq:V2_smooth} -- \eqref{eq:V2_normalized} and $V_3$ satisfying assumptions \eqref{eq:V3_smooth} -- \eqref{eq:V3_angle72}, all ground states $(X_{\min}, y_{\min})$ of $V_L^{per} (\cdot, \cdot)$, where $y_{\min}$ is $L$-periodic, satisfy
\begin{equation}
\label{eq:min_dist_periodic}
	\min_{x \neq x'} |y_{\min}(x') - y_{\min}(x)| > 1 - \alpha.
\end{equation}
\end{lemma}

\begin{proof}
	This lemma actually follows from the more general Lemma 7.1 in \cite{e09}. As with Lemma~\ref{lem:min_dist}, the proof is slightly simpler for our potential, so we include it here.

Without loss of generality, assume that $L \geq 3$. The cases $L' \in \{ 1, 2 \}$ can be treated by $L = 3 L'$. Define $M := \max_{\eta \in \mathbb{R}^2} \# \{ x : y(x) \in B(\eta, \frac{1}{2} (1-\alpha)) \}$. We wish to show that $M=1$.

We assume without loss of generality that $\eta = \frac{L}{2} a_1 + \frac{L}{2} a_2$, i.e.\ $\eta$ is at the center of the reference cell $LU$. Define $B_M := B(\eta, \frac{1}{2} (1-\alpha))$ and $\mathcal{A} := y^{-1}(B_M)$. Because $B_M$ is centered in the reference cell $LU$ and $L \geq 3$, $\mathcal{A} \subset X \cap LU$. Therefore, using \eqref{eq:energy_change_A} for the pair energy, the total change in energy, which must be non-negative since we are in a global minimizer, is
\begin{equation*}
	- \frac{1}{2} e(\mathcal{A}, \mathcal{A})
	- \frac{1}{2} e(\mathcal{A}, (\Pi \mathcal{A}) \setminus \mathcal{A})
	- e(\mathcal{A}, X \setminus (\Pi \mathcal{A}))
	- \frac{1}{2} \sum_{\substack{x \in X \cap LU, x', x'' \in X,\\ \{x, x', x''\}  \cap (\Pi \mathcal{A}) \neq \emptyset\\ x, x', x'' \text{ distinct}}} a(x, x', x'')
	\geq 0.
\end{equation*}

So, we have
\begin{equation}
\label{eq:min_dist_ineq_periodic}
	- e(\mathcal{A}, \mathcal{A})
	- e(\mathcal{A}, (\Pi \mathcal{A}) \setminus \mathcal{A})
	- \sum_{\substack{x \in X \cap LU; x', x'' \in X,\\ \{x, x', x''\}  \cap (\Pi \mathcal{A}) \neq \emptyset\\ x, x', x'' \text{ distinct}}} a(x, x', x'')
	\geq  2 e(\mathcal{A}, X \setminus (\Pi \mathcal{A})).
\end{equation}

As in the proof of Lemma~\ref{lem:min_dist}, we have
\begin{equation*}
e(\mathcal{A}, \mathcal{A}) = \sum_{x, x' \in \mathcal{A}} e(\{x, x'\}) \geq \frac{1}{\alpha} M (M-1).
\end{equation*}

Also, since $f$ and $V_3$ are non-negative, we have
\begin{equation*}
\sum_{\substack{x \in X \cap LU, x', x'' \in H,\\ \{x, x', x''\}  \cap (\Pi \mathcal{A}) \neq \emptyset\\ x, x', x'' \text{ distinct}}} a(x, x', x'') \geq 0.
\end{equation*}

Now, we estimate the interaction between the particles in $\mathcal{A}$ and their periodic images. We have
\begin{align*}
e(\mathcal{A}, (\Pi \mathcal{A}) \setminus \mathcal{A})
= \sum_{x \in \mathcal{A}} \sum_{x' \in (\Pi \mathcal{A}) \setminus \mathcal{A}} e(\{x, x'\})
& = \sum_{x \in \mathcal{A}} \sum_{\xi \in L \sqrt{3} A_2 \setminus \{ 0 \} } \sum_{x' \in \mathcal{A} } e(\{x, x' + \xi \})
\end{align*}
By the periodic boundary conditions, for $\xi \in L \sqrt{3} A_2 \setminus \{ 0 \} $, $y(x' + \xi) = y(x') + \xi$. Therefore,
\begin{equation*}
	e(\{x, x' + \xi \}) = V_2(|y(x) - y(x' + \xi)|) = V_2(|y(x) - y(x') - \xi|).
\end{equation*}
Since $x, x' \in \mathcal{A}$, we have $y(x), y(x') \in B_M$, and $|y(x) - y(x')| \leq 1 - \alpha \leq 1$. Then, by the reverse triangle inequality, for $L \geq 3$,
\begin{equation*}
	|y(x) - y(x') - \xi| \geq |\xi| - |y(x) - y(x')| \geq L - (1 - \alpha) \geq L - 1 \geq 2.
\end{equation*}
Because of assumption \eqref{eq:V2_long_range}, $V_2(|y(x) - y(x' + \xi)|) = 0$. Thus, we have
\begin{align}
\label{eq:periodic_interaction}
e(\mathcal{A}, (\Pi \mathcal{A}) \setminus \mathcal{A})
& = \sum_{x \in \mathcal{A}} \sum_{\xi \in L \sqrt{3} A_2 \setminus \{ 0 \} } \sum_{x' \in \mathcal{A} } e(\{x, x' + \xi \}) = 0.
\end{align}

Combining these results, the inequality~\eqref{eq:min_dist_ineq_periodic} becomes
\begin{equation*}
	-\frac{1}{\alpha} M (M-1) \geq 2 e(\mathcal{A}, X \setminus (\Pi \mathcal{A}),
\end{equation*}
similar to the inequality~\eqref{eq:ineq2_min_dist} in the proof of Lemma~\ref{lem:min_dist}.

As in that proof, we can argue that there exists $C>0$, independent of $\alpha_0$, such that
\begin{equation*}
	e(\mathcal{A}, X \setminus (\Pi \mathcal{A})) = \sum_{x \in \mathcal{A}, x' \in X \setminus (\Pi \mathcal{A})} e(\{x, x'\}) \geq -CM^2.
\end{equation*}

Therefore,
\begin{equation*}
	-\frac{1}{\alpha} M (M-1) \geq -2 C M^2,
\end{equation*}
which is equivalent to
\begin{equation*}
	(1 - 1/M) \leq 2 C \alpha.
\end{equation*}
Taking $\alpha_0 = \frac{1}{4C}$, if this inequality holds for $\alpha < \alpha_0$, then we must have $M=1$. This proves the result.
\end{proof}

Now, we state and prove the periodic version of the result that regular atoms do not have mid-range interactions.
\begin{lemma}
\label{lem:rigidity_periodic}
There exists $\alpha_0 \in \left( 0, \frac{1}{4} \right)$ and $\epsilon > 0$, such that
\begin{itemize}
	\item for all $\alpha$ with $0 < \alpha < \alpha_0$,
	\item for all potentials $V$ of form~\eqref{eq:energy_periodic_subset}, with $V_2$ satisfying assumptions \eqref{eq:V2_smooth} -- \eqref{eq:V2_normalized} and $V_3$ satisfying assumptions \eqref{eq:V3_smooth} -- \eqref{eq:V3_angle72},
	\item for all $L \in \mathbb{N}$,
	\item for all ground states $(X_{\min}, y_{\min})$ of $V_L^{per}(\cdot)$, where $y_{\min}$ is $L$-periodic, and
	\item for all $x \in G_\epsilon, x' \in X_{\min}$ with $x' \neq x$,
\end{itemize}
we have
\begin{itemize}
	\item $|y(x') - y(x)| \leq 1 + \alpha$, or
	\item $|y(x') - y(x)| \geq \frac{3}{2}$.
\end{itemize}
\end{lemma}

\begin{proof}
As in the proof of Lemma~\ref{lem:min_dist_periodic}, assume without loss of generality that $L \geq 3$. Let $\alpha_0 > 0$ be such that the inequality~\eqref{eq:min_dist_periodic} in Lemma~\ref{lem:min_dist_periodic} holds. Let $V$ satisfy the given assumptions, and let $y : X \rightarrow \mathbb{R}^2$ be a ground state. Let $x \in G_\epsilon$. To show that all particles are either closer than $1 + \alpha$ or further than $\frac{3}{2}$ away from $x$, we will use a proof by contradiction. Suppose that $p$ is such that $1 + \alpha < |y(p) - y(x)| < \frac{3}{2}$. As in the proof of Lemma~\ref{lem:rigidity}, we can show that $p$ is part of a triple $(x, p, p')$ with a bond angle less than $\frac{2\pi}{5}$. We now show that this contradicts the fact that the configuration is a ground state.

The change in energy upon removing the set $\Pi x$ must be non-negative since we are in a ground state. Using \eqref{eq:energy_change_x} for the pair energy, the total change in energy is
\begin{equation}
\label{eq:energyChange_periodic}
	-\frac{1}{2} e \left(x, \Pi x \setminus \{x\} \right)
	- e \left(x, X \setminus (\Pi x) \right)
	- \frac{1}{2} \sum_{\substack{\bar{x} \in X \cap LU, x', x'' \in X,\\ \{\bar{x}, x', x''\}  \cap (\Pi x) \neq \emptyset\\ \bar{x}, x', x'' \text{ distinct}}} a(\bar{x}, x', x'') \geq 0.
\end{equation}

As in Equation~\eqref{eq:periodic_interaction} from the proof of Lemma~\ref{lem:min_dist_periodic}, for the interaction between $x$ and its periodic images, we have
\begin{equation}
\label{eq:rigidity_image_periodic}
	e \left(x, \Pi x \setminus \{x\} \right) = 0.
\end{equation}

As in Equation~\eqref{eq:rigidity_edge} from the proof of Lemma~\ref{lem:rigidity}, we have
\begin{equation}
\label{eq:rigidity_edge_periodic}
	e \left(x, X \setminus (\Pi x) \right) \geq -3 - C\alpha.
\end{equation}

Also, since the angle energy is non-negative, and there is at least one bond angle $\theta_{p, x, p'} \leq \frac{2\pi}{5}$ with $|y(p) - y(x)| < \frac{3}{2}$ and $|y(p') - y(x)| < \frac{3}{2}$, by assumption~\eqref{eq:V3_angle72}, we have
\begin{equation}
\label{eq:rigidity_angle_periodic}
	\frac{1}{2} \sum_{\substack{\bar{x} \in X \cap LU, x', x'' \in X,\\ \{\bar{x}, x', x''\}  \cap (\Pi x) \neq \emptyset\\ \bar{x}, x', x'' \text{ distinct}}} a(\bar{x}, x', x'')
	\geq \frac{1}{2} (a(x, p, p') + a(x, p', p)) = V_3(\theta_{p, x, p'}) \geq V_3 \left( \frac{2\pi}{5} \right) \geq 4.
\end{equation}

Therefore, substituting \eqref{eq:rigidity_image_periodic}, \eqref{eq:rigidity_edge_periodic}, and \eqref{eq:rigidity_angle_periodic} into \eqref{eq:energyChange_periodic}, we have
\begin{equation*}
	0 \leq 0 + \left( 3 + C\alpha \right) - 4 = -1 + C\alpha.
\end{equation*}

Choosing $\alpha_0$ to be the lesser of the previous value and $\frac{1}{2C}$, we have for $0 < \alpha < \alpha_0$ that $-\frac{1}{2} > 0$, which is a contradiction. Thus, for a ground state, there can be no such point $p$ with $1 + \alpha \leq |y(p) - y(x)| \leq \frac{3}{2}$. This proves the claim.
\end{proof}

\subsection{Honeycomb lattice formation using the main estimate and comparison with identity map}
We are now able to prove that the ground state of the energy subject to periodic boundary conditions is a honeycomb lattice.

\begin{proof}[Proof of Theorem~\ref{thm:periodic}]
Using Lemmas~\ref{lem:min_dist_periodic} and~\ref{lem:rigidity_periodic}, we get an estimate similar to \eqref{eq:main_ineq}, proceeding by a similar argument. The main difference is that we must replace $G_{\epsilon}$ and $\mathcal{M}$ by their equivalence classes $\widetilde{G}_{\epsilon}$ and $\widetilde{\mathcal{M}}$, and we replace $N = \# X_N$ by $\# \widetilde{X}$.

Using the same steps as in Section~\ref{sec:energy_decomp}, we estimate the energy as
\begin{align}
\label{eq:energy_decomp_periodic}
	V_L^{per}(X, \{y\} )
	&\geq \sum_{x \in G_\epsilon \cap LU} V_{\mathcal{N}} (x)
	+ \sum_{x \in G_\epsilon^c \cap LU} V_{\mathcal{N}} (x)
	+ \frac{1}{2} \sum_{x \in X \cap LU}\ \sum_{\substack{x' \notin \mathcal{N}(x)\\ x' \neq x}} e(\{ x, x' \}).
\end{align}

For the mid-range interactions, we want the equivalent of Lemma~\ref{lem:mid-range}:
\begin{equation*}
	\frac{1}{2} \sum_{x \in X \cap LU}\ \sum_{\substack{x' \notin \mathcal{N}(x)\\ x' \neq x}} e(\{ x, x' \})
	\geq - \alpha \# \widetilde{\mathcal{M}}.
\end{equation*}
This will be true as long as
\begin{equation*}
	\# \widetilde{\mathcal{M}} = \frac{1}{2} \sum_{x \in X \cap LU}\ \sum_{\substack{x' \in X\\ 1 + \alpha < |y(x')-y(x)| < \frac{3}{2}}} 1.
\end{equation*}

If $x \in X \cap LU$, $x' \in X$, and $1 + \alpha < |y(x')-y(x)| < \frac{3}{2}$, then $\{ x, x' \} \in \mathcal{M}$, so the sum will only count elements of $\mathcal{M}$, which can be viewed as representatives of $\widetilde{\mathcal{M}}$. Now, consider an equivalence class $[\omega] \in \widetilde{\mathcal{M}}$. We have two different cases.

\textbf{Case 1:} $[\omega]$ has a representative with both elements in $X \cap LU$. This representative, and therefore the equivalence class, will be counted exactly twice in the sum.

\textbf{Case 2:} $[\omega]$ has no representative with both elements in $X \cap LU$. Then, $[\omega]$ has two distinct representatives that have exactly one element in $X \cap LU$. To see this, suppose one of these representatives is $\{x, x'\}$ with $x \in X \cap LU$ and $x' \in X \setminus (X \cap LU)$. Then, since $X$ is $L$-periodic, there exists a unique $\tau \in L \sqrt{3} A_2 \setminus \{0\}$ such that $x' + \tau \in X \cap LU$. We have $x + \tau \in X$, and $\{x + \tau, x' + \tau \} \sim \{x, x'\}$. If $x' + \tau = x$, we need to ensure that $x + \tau \neq x'$, so that these are not the same pair. But $x + \tau = x' + 2\tau \neq x'$, so $\{x + \tau, x' + \tau \} \neq \{x, x'\}$. Since the sum counts both of the representatives that have exactly one element in $X \cap LU$, the equivalence class is counted twice.

We see that in either case, the equivalence class will be counted twice in the sum, so that a factor of $\frac{1}{2}$ is needed to count $\widetilde{\mathcal{M}}$.

The same argument can be used to show
\begin{align*}
	\# \widetilde{\mathcal{M}}_d
	&= \frac{1}{2} \sum_{x \in G_{\epsilon}^c \cap LU} \sum_{\substack{x' \in G_{\epsilon}^c\\ 1+\alpha < |y(x')-y(x)| < \frac{3}{2}}} 1,
\end{align*}
from which we can derive
\begin{equation*}
	\# \widetilde{\mathcal{M}} = \# \widetilde{\mathcal{M}}_d \leq C (\# \widetilde{X} - \# \widetilde{G}_{\epsilon}),
\end{equation*}
as in Lemma~\ref{lem:ub_midrange}.

Using the estimate on non-neighbor pairs in \eqref{eq:energy_decomp_periodic}, we have
\begin{align}
	V_L^{per}(X, \{y\} )
	&\geq \sum_{x \in G_\epsilon \cap LU} V_{\mathcal{N}} (x)
	+ \sum_{x \in G_\epsilon^c \cap LU} V_{\mathcal{N}} (x)
	- \alpha \# \widetilde{\mathcal{M}}.
\end{align}

Using Lemmas~\ref{lem:lb_all} and~\ref{lem:lb_defected} to estimate the neighborhood energies and then simplifying, this becomes
\begin{align*}
	V &\geq	
	\left( -\frac{3}{2} + 3 V_3(2\pi/3) \right) \# (X \cap LU) + \Delta_\epsilon (\# (X \cap LU) - \# (G_{\epsilon} \cap LU))\\
	&\qquad \qquad + \sum_{x \in G_\epsilon \cap LU} \left( W_{e}(x) + W_{a}(x) \right) - \alpha \# \widetilde{\mathcal{M}}\\
	&= \left( -\frac{3}{2} + 3 V_3(2\pi/3) \right) \# \widetilde{X}+ \Delta_\epsilon (\# \widetilde{X} - \# \widetilde{G}_{\epsilon}) + \sum_{x \in \widetilde{G}_{\epsilon}} \left( W_{e}(x) + W_{a}(x) \right) - \alpha \# \widetilde{\mathcal{M}},
\end{align*}
where we have used $\#(X \cap LU) = \# \widetilde{X}$ and $\# (G_{\epsilon} \cap LU) = \# \widetilde{G}_{\epsilon}$.

Using the estimate on the number of mid-range pairs, there exists $\epsilon > 0$ such that this becomes
\begin{align*}
	V &\geq
	\left( -\frac{3}{2} + 3 V_3(2\pi/3) \right) \# \widetilde{X} + \Delta_\epsilon (\# \widetilde{X} - \# \widetilde{G}_{\epsilon}) + \sum_{x \in \widetilde{G}_{\epsilon}} \left( W_{e}(x) + W_{a}(x) \right) - C \alpha (\# \widetilde{X} - \# \widetilde{G}_{\epsilon}).
\end{align*}

Choosing $\alpha_0 = \frac{1}{2C}\Delta_\epsilon$ so that $C \alpha < \frac{1}{2}\Delta_\epsilon$, as before, we have
\begin{equation}
\label{eq:main_ineq_periodic}
\begin{split}
	V(X_{\min}, y_{\min}) & \geq
	\left( -\frac{3}{2} + 3 V_3(2\pi/3) \right) \# \widetilde{X} + \frac{\Delta_\epsilon}{2} (\# \widetilde{X} - \# \widetilde{G}_{\epsilon})\\
	& \qquad + \sum_{x \in \widetilde{G}_\epsilon} \left( \sum_{j=1}^3 \frac{1}{2} \left| r_j - 1 \right|^2 + \sum_{j=1}^3 \frac{m}{2} \left| \theta_j - \frac{2\pi}{3} \right|^2 \right).
\end{split}
\end{equation}

We compare to the energy of the competitor $X = H$, $y(x) = x$, which has energy
\begin{equation*}
	V(X, y) = \left( -\frac{3}{2} + 3 V_3(2\pi/3) \right) \left(2 L^2 \right).
\end{equation*}
(Recall that the set $H \cap LU$ contains $2 L^2$ particles.) Thus, we obtain an inequality for the minimum energy:
\begin{equation*}
	\left( -\frac{3}{2} + 3 V_3(2\pi/3) \right) \left( 2 L^2 \right) \geq V(X_{\min}, y_{\min}).
\end{equation*}

Combining this with \eqref{eq:main_ineq_periodic}, we have
\begin{align*}
	0 &\geq
	\left( -\frac{3}{2} + 3 V_3(2\pi/3) \right) (\# \widetilde{X} - 2 L^2) + \frac{\Delta_\epsilon}{2} (\# \widetilde{X} - \# \widetilde{G}_{\epsilon})\\
	& \qquad \qquad + \sum_{x \in \widetilde{G}_\epsilon} \left( \sum_{j=1}^3 \frac{1}{2} \left| r_j - 1 \right|^2 + \sum_{j=1}^3 \frac{m}{2} \left| \theta_j - \frac{2\pi}{3} \right|^2 \right)
\end{align*}

Since $-\frac{3}{2} + 3 V_3(2\pi/3) < 0$ by assumption~\eqref{eq:V3_angle120} and $\# \widetilde{X} \leq 2 L^2$, the first term will be positive unless $\widetilde{X} = 2 L^2$. Since the other terms are non-negative, this implies that there are $2 L^2$ atoms. If the first term is zero, the other terms must be zero in order to satisfy the inequality. Therefore, there are zero defected atoms (i.e.\ all atoms have 3 neighbors), and all bond lengths are 1 and all bond angles are $\frac{2\pi}{3}$. A configuration with these properties must be a honeycomb lattice.
\end{proof}

\appendix
\section{Appendix: Brenner potential}
\label{sec:brenner}
Since this article is motivated by the use of the Brenner potential to model interactions in carbon nanostructures such as graphene and carbon nanotubes, we review its definition here. The Brenner potential \cite{brenner90} is
\begin{equation}
	V = \frac{1}{2} \sum_{x \in X} \sum_{\substack{x' \in X\\ x' \neq x}} V_R(|y(x') - y(x)|) - \bar{B}_{x, x'} V_A(|y(x') - y(x)|).
\end{equation}
This is not simply a pair potential, since the bond order $\bar{B}_{x, x'}$ depends on triples.

The repulsive and attractive potentials are
\begin{align}
	V_R(r) &= f(r) \frac{D}{S-1} \exp \left( -\sqrt{2 S} \beta (r - R_e) \right), \text{ and} \\
	V_A(r) &= f(r) \frac{D S}{S-1} \exp \left( -\sqrt{2/S} \beta (r - R_e) \right).
\end{align}
$R_e$ is the equilibrium bond length when $\bar{B} = 1$ and $D$, $S$, and $\beta$ are parameters that control the shape of the potential.

The range of the potential is limited by a cutoff function $f$:
\begin{equation}
	f(r) =
	\begin{cases}
	1 & \text{if $r < R_1$}\\
	\frac{1}{2} \left( 1 + \cos( \pi (r - R_1)/(R_2 - R_1) ) \right) & \text{if $R_1 < r < R_2$}\\
	0 & \text{if $R_2 < r$}.
	\end{cases}
\end{equation}
The variable $r$ represents the distance between atoms, and the parameters $R_1$ and $R_2$ define the interval where the function decreases from 1 to 0.

In addition, the Brenner potential involves the bond order\footnote{We consider the simplified form of the bond order used in other work, such as \cite{shibuta03} and \cite{zhang02}.} $\bar{B}$, which involves triples of carbon atoms. The bond order involves the parameters $\delta$, $a_0$, $c_0$, and $d_0$. The bond order $\bar{B}_{x, x'} = \frac{1}{2} ( B_{x, x'} + B_{x', x} ) $, where
\begin{equation}
	B_{x, x'} = \left( 1 + \sum_{\substack{x'' \in X\\ x \neq x'' \neq x'}} G(\theta_{x', x, x''}) f(|y(x'') - y(x)|) \right)^{- \delta}.
\end{equation}

The angular function $G$ is defined as
\begin{equation}
	G(\theta) = a_0 \left( 1 + \frac{c_0^2}{d_0^2} - \frac{c_0^2}{d_0^2 + (1 + \cos \theta)^2} \right).
\end{equation}
$G(\theta)$ takes its minimum value when $\theta = \pi$, which in turn gives a greater value for $B_{ij}$ and a lesser value for $V$.

The first set of parameter values given in \cite{brenner90} is $D = 6.325$~eV, $S = 1.29$, $\beta = 1.5$~\AA\textsuperscript{-1}, $R_{e} = 1.315$~\AA, $R_{1} = 1.7$~\AA, $R_{2} = 2.0$~\AA, $\delta = 0.80469$, $a_0 = 0.011304$, $c_0 = 19.0$, and $d_0 = 2.5$. The second set of parameter values is $D = 6.0$~eV, $S = 1.22$, $\beta = 2.1$~\AA\textsuperscript{-1}, $R_{e} = 1.39$~\AA, $R_{1} = 1.7$~\AA, $R_{2} = 2.0$~\AA, $\delta = 0.5$, $a_0 = 0.00020813$, $c_0 = 330$, and $d_0 = 3.5$.

For a honeycomb lattice, we find that the lattice parameter minimizing the Brenner potential is
\begin{equation*}
	r = R_e - \frac{1}{\beta} \frac{\sqrt{S/2}}{S-1} \ln B_0,
\end{equation*}
where $B_0 = \left( 1 + 2 G(\frac{2\pi}{3}) \right)^{-\delta}$ is the bond order of any bond in a honeycomb lattice when the lattice parameter is less than the cutoff $R_1$.

\textit{Acknowledgements.} B.~Farmer was supported under NSF grant DMS-1317730 and US-DOE grant DE-SC0012733. S.~Esedo\={g}lu was supported under NSF grant DMS-1317730. P.~Smereka was supported under NSF grant DMS-1417053. B.~Farmer and S.~Esedo\={g}lu express their gratitude to their collaborator, mentor, and friend Peter Smereka, who passed away during the final stages of preparation of this manuscript.

\bibliographystyle{plain}
\bibliography{references}

\end{document}